\documentclass[10pt, a4paper,reqno]{amsart}
\usepackage{latexsym}
\usepackage{amsmath}
\usepackage[dvips]{graphicx}%
\usepackage{amsfonts}%
\usepackage{amssymb}
\usepackage{color}
\newcommand{\curl}{\operatorname*{curl}}
\newcommand{\sym}[1]{#1_{sym}}

\newcommand{\B}{\mathbf{B}}
\newcommand{\beq}{\begin{equation}}
\newcommand{\eeq}{\end{equation}}

\newcommand{\bu}{\mathbf{b}}

\newcommand{\md}{\mathcal{M}_{dis}}
\newcommand{\ro}{\varrho}

\def\R{\mathbb R}
\def\N{\mathbb N}
\def\C{\mathbb C}
\def\Z{\mathbb Z}
\def\cal{\mathcal}

\def\H{{\cal H}}

\def\de{\delta}
\def\e{\varepsilon}

\def\s{\sigma}
\def\vphi{\varphi}
\def\om{\omega}
\newcommand{\medint}{-\kern -,375cm\int}
\newcommand{\medintinrigo}{-\kern -,315cm\int}
\newcommand{\M}{\mathbb{M}^{2\times2}}

\def\Om{\Omega}

\newcommand{\wto}{\rightharpoonup}
\def\pa{\partial}

\def\Div{{\rm div}}

 \newcommand{\Hc}[1]{\mathbf{H}_\#(\curl; \Om_#1; \mathbb{M}^{2\times2})}
\numberwithin{equation}{section}
\newtheorem{theorem}{Theorem}[section]

\newtheorem{corollary}[theorem]{Corollary}

\newtheorem{lemma}[theorem]{Lemma}

\newtheorem{proposition}[theorem]{Proposition}

\theoremstyle{definition}
\begingroup
\newtheorem{definition}[theorem]{Definition}
\newtheorem{remark}[theorem]{Remark}

\endgroup

\begin{document}

\title[Dislocations in elastic films]{A model for dislocations in epitaxially strained elastic films}
\author{I. Fonseca, N. Fusco, G. Leoni, M. Morini}
\address[I.\ Fonseca]{Department of Mathematical Sciences, Carnegie Mellon University, Pittsburgh, PA, U.S.A.}
\email[I.\ Fonseca]{fonseca@andrew.cmu.edu}
\address[N.\ Fusco]{Dipartimento di Matematica e Applicazioni "R. Caccioppoli",
Universit\`{a} degli Studi di Napoli "Federico II" , Napoli, Italy}
\email[N.\ Fusco]{n.fusco@unina.it}
\address[G. \ Leoni]{Department of Mathematical Sciences, Carnegie Mellon University, Pittsburgh, PA, U.S.A.}
\email[G. \ Leoni]{giovanni@andrew.cmu.edu}
\address[M.\ Morini]{Dipartimento di Matematica,
Universit\`{a} degli Studi di Parma , Parma, Italy}
\email[M.\ Morini]{massimiliano.morini@unipr.it}

\begin{abstract}
A  variational model for epitaxially strained films accounting for the presence of dislocations is considered. Existence, regularity and some qualitative properties of solutions are addressed. 
\end{abstract}

\maketitle

\tableofcontents

\section{Introduction}

The ability to control the morphology of elastically stressed thin films is paramount in the manufacturing of  microelectronics and optical devices. 
Due to the misfit between the film and the substrate lattice constants, the film may undergo a morphological change, known as the Asaro-Grinfeld-Tiller (AGT) instability (see \cite{AsTi}, \cite{Gr0}).   This is a stress relief mechanism, by which the system decreases the elastic energy by allowing non-planar morphologies when a critical thickness is achieved.  Such  threshold effect is usually explained as the result of two competing  forms of energy: the surface energy, which favors flat configurations, and the bulk elastic energy, which in turn is decreased by wavy or corrugated configurations.

An extensive literature is devoted to the modeling and to the numerical analyis of strained epitaxial films; see for instance \cite{GaNi}, \cite{Spencer}, \cite{SpMe},  \cite{ST} and the references therein.  
Several variational models have been proposed to study epitaxial growth, both in  the static case (see \cite{BeGoZw, Bo0, Bo, BC, BoFG,  FFLM, FuMo,GoZw}) as well as in the time-dependent setting (see \cite{FFLM2, FFLM3, Pi}), starting with the free-energy approach of \cite{Grinfeld}.

Experiments indicate that the nucleation of dislocations is a further  mode of strain relief (in addition to the already mentioned profile buckling) for sufficiently thick films (see, for instance,  \cite{EPBSG, GaNi, HMRG, JPBH, TeLe}). Indeed, when a cusp-like morphology is formed, the resulting local stress at a surface valley has a greater energy than that produced by the nucleation of a dislocation.  Once the dislocation is formed, it migrates to the film/substrate interface, and  the film surface relaxes  towards a  planar-like  morphology.

In this paper we propose a mathematical model, which takes into account the formation of misfits dislocations. We start by recalling the variational formulation studied in  \cite{BC} and \cite{FFLM} (see also \cite{BCS} and \cite{CS}) within the context of equilibrium configurations of epitaxially strained films without dislocations. As in those papers we work within the theory of linear elasticity. We consider two-dimensional configurations, corresponding to three-dimensional morphologies with planar symmetry.   The reference configuration of the film is described as
$$
\Omega_h:=\left\{  z=\left(  x,y\right)\in\R^2  :\,0<x<\ell,\,0<y<h\left(
x\right)  \right\}\,,
$$
where the function $h: [0,\ell]  \rightarrow\left[  0,\infty\right)$ represents the free-profile of the film. 
The vector field  $u:\Omega_h\rightarrow\mathbb{R}^{2}$ represents the displacement of the film and 
\[
{E}\left({u}\right)  :=\frac{1}{2}\left(  \nabla{u}+\nabla^T{u}\right)
\]
its strain. The presence of a mismatch between the lattice constants of the film and the substrate  is incorporated in the model by  prescribing a Dirichlet boundary condition of the form $u(x,0)=(e_0x,0)$ at the interface, with $e_0\neq 0$. This corresponds to the case of  a film growing on an infinitely rigid substrate.  

 As customary in the physical literature, we also require  the periodicity conditions $h(0)=h(\ell)$ and $\nabla u(0,y)=\nabla u(\ell,y)$. The energy associated with a dislocation-free configuration $(h,u)$, when $h$ is smooth,  
 is given by
 $$
 \mathcal{G}(h,u):=\int_{\Om_h}\Bigl[\mu|E(u)|^2+\frac{\lambda}{2}({\rm div} u)^2\Bigr]\, dz+ \gamma \H^1(\Gamma_h)\,,
 $$
 where $\mu$ and $\lambda$ are the Lam\'e coefficients of the material, $\gamma$ is the surface tension on the profile of the film, 
 $\Gamma_h$ denots  the graph
 of $h$, and $\H^1$ stands for the one-dimensional Hausdorff measure. 
 
 Equilibrium configurations corresponf to local or global minimizers of $G$ among all admissible configurations, with prescribed volume. Notice that smooth   minimizing sequences may converge to   irregular configurations, with the profile $h$ being a lower semicontinuous function of bounded variation.  In particular, vertical parts and cuts may appear in the (extended) graph of $h$. This requires extending the definition of $G$ to a larger class of possibly irregular reachable configurations, through a relaxation procedure. 
This has been done in \cite{BC} and \cite{FFLM} (see also \cite{BCS} and \cite{CS}), and it leads to the relaxed energy:
 \begin{equation}\label{Frelax}
 \mathcal{G}(h,u)=\int_{\Om_h}\Bigl[\mu|E(u)|^2+\frac{\lambda}{2}({\rm div} u)^2\Bigr]\, dz+ \gamma \H^1(\Gamma_h)
 +2 \gamma\H^1(\Sigma_h)\,,
 \end{equation}
where $\Sigma_h$ is the set of \emph{vertical cuts} defined as
$$
\Sigma_h:=\{(x,y):\, x\in [0,\ell),\, h(x)<y<\min\{h(x-), h(x+)\}\}\,,
$$
 with $h(x\pm)$ denoting the right and left limit at $x$.  Note that the factor 2 appearing in the last term of \eqref{Frelax} is due to the fact that in the approximation procedure vertical cuts result from the collapsing of needle-like smooth profiles into a segment whose length in the limit is  counted twice.

Next we modify $\mathcal{G}$ to account for the presence of isolated misfit dislocations in the film.  The mathematical modeling of dislocations has been studied by several authors;  see for instance \cite{ADGP, ArOr, BFLM, CGO, DGP, FPP, GPPS, HuOr1, HuOr2, MPS, Po}, and the references therein.

Volterra's  dislocations may be viewed as topological point singularities of the field (see \cite{N}). To be precise, given a set of points 
$\{z_1, \dots, z_k\}\subset\Om_h$ and a set of vectors $\{\bu_1, \dots, \bu_k\}\subset\R^2$, a \emph{strain field} $H$ is compatible with a system of dislocations located at $z_1, \dots, z_k$ and having \emph{Burgers vectors} $\bu_1, \dots, \bu_k$ if 
\begin{equation}\label{curl}
\curl H=\sum_{i=1}^k\bu_i \de_{z_i}\,,
\end{equation}
where $\de_z$ denotes the Dirac delta at $z$. Since the elastic continuum model is not valid near the singularities,  some kind of regularization is needed. A standard approach in the engineering literature (see \cite{N})  is to remove a core $B_{r_0}(z_i)$ of radius $r_0>0$ around each dislocation and associate with $H$ the (finite) elastic energy
$$
\int_{\Om_h\setminus \cup_{i=1}^kB_{r_0}(z_i)}\Bigl[\mu|H_{sym}|^2+\frac{\lambda}{2}(\mathrm{tr}(H))^2\Bigr]\, dz\,,
$$ 
where $H_{sym}:=(H+H^T)/2$. The mathematical study of this energy can be found, e.g., in \cite{CeLe, DGP, GLP, MPS}. 

In this paper, following \cite{HMRG}, we consider a variant of this approach, which consists in regularizing the {\em dislocation measure} $\s:=\sum_{i=1}^k\bu_i \de_{z_i}$ through  a convolution procedure. To be precise, we replace \eqref{curl} with the compatibility condition 
\begin{equation}\label{compa}
\curl H=\sigma*\ro_{r_0}\,,
\end{equation}
where $\ro_{r_0}:=(1/r_0^2)\ro(\cdot/r_0)$ is a convolution kernel, with $\ro$ a standard mollifier compactly supported in the unit ball. 
 Here $r_0>0$ is a fixed constant that  may be interpreted as before as the core radius.  Since the set of strain fields $H$ satisfying condition \eqref{compa} and with finite energy, i.e.,
 \begin{equation} \label{finite energy}
 \int_{\Om_h}\Bigl[\mu|H_{sym}|^2+\frac{\lambda}{2}(\mathrm{tr}(H))^2\Bigr]\, dz<+\infty
 \end{equation}
 is non-empty, for any given profile $h$ and any given dislocation measure $\sigma$, the compatible strain field $H$ minimizing the elastic energy \eqref{finite energy} is well defined and satisfies the div-curl system
 \begin{equation}\label{divcurl}
 \left\{
 \begin{array}{lr}
 \curl H=\sigma*\ro_{r_0}\vspace{-5pt} &  \\
 \vspace{-5pt} & \quad \text{in $\Om_h$.}\\
 \mu\, \Div\, H+(\lambda+\mu)\nabla(\mathrm{tr}(H))=0\,
 \end{array}
 \right.
 \end{equation}
 Note that the above system admits an equivalent formulation in terms of the so-called Airy stress function $w$ associated with $H$ through the identity
 $$
 \nabla^2 w=\frac12
\left(
\begin{array}{cc}
(2\mu+\lambda)H_{22}+\lambda H_{11} &- \mu (H_{12}+H_{21})   \\
- \mu (H_{12}+H_{21})  & (2\mu+\lambda)H_{11}+\lambda H_{22}  \\
\end{array}
\right)\,,
 $$
 see \cite[Chapter~12]{Fung}. Indeed,  \eqref{divcurl} can be rewritten as (see \cite{HMRG})
 $$
 \Delta^2w=\curl (\sigma*\ro_{r_0})\qquad\text{in $\Om_h$}\,.
 $$
 
 Adopting the above convolution-based regularization, the total energy associated with a profile $h$, a dislocation measure $\s$, and a strain field $H$, satisfying the compatibility conditions \eqref{compa}, is given by
 \begin{equation}\label{total}
 F(h,\s, H):=\int_{\Om_h}\Bigl[\mu|H_{sym}|^2+\frac{\lambda}{2}(\mathrm{tr}(H))^2\Bigr]\, dz+\gamma\H^1(\Gamma_h)+2\gamma
 \H^1(\Sigma_h)\,.
\end{equation}
In Section \ref{sec:nonnucl} we assume that a finite number $k$ of dislocations, with 
 given Burgers vectors $\mathbf{B}:=\{\bu_1, \dots, \bu_k\}\subset\R^2$, are already present in the film, and we address the problem of finding the optimal configuration, i.e.,  the profile $h$ and the location $z_1, \dots, z_k$ of the $k$ dislocations which  minimize the total energy, under a given volume constraint $|\Om_h|=d$. To be precise, denoting by $ 
 X(e_0; \mathbf{B})$ the set of admissible triples $(h,\sigma,H)$, in Theorem \ref{th:existence} below, we prove 
 \begin{theorem}
 The minimization problem 
 \beq\label{minG1}
 \min\{F(h,\sigma, H):\, (h, \sigma, H)\in X(e_0; \mathbf{B}), \, |\Om_h|=d\}\,.
 \eeq
 admits a solution.
 	 \end{theorem}
 We then show that the equilibrium profile $h$ obtained above satisfies the same regularity properties proved in \cite{FFLM} (see also \cite{DF, FuMo})  in the dislocation-free case. Namely,   
 \begin{theorem} Let  $(\bar h,\bar \sigma, H_{\bar h,\bar\s})\in  X(e_0;\B)$ be a  minimizer of \eqref{minG1}. Then $\bar h$ has at most finitely many cusp points and vertical cracks, its graph is of class $C^1$ away from this finite set, and of class $C^{1,\alpha}$, $\alpha\in (0,\frac{1}{2})$ away from this finite set and off the substrate.
 \end{theorem}
 For a more detailed qualitative description of this regularity result we refer to Theorem \ref{th:regolarita} below. The overall strategy to prove this theorem is the same used in  \cite{FFLM}. However, there are many new technical issues due to the presence of dislocations, which require new ideas. In particular, a major difficulty arises in showing that  the volume constraint can be replaced by a volume penalization. In the dislocation-free case this was based on a straightforward truncation argument, which fails in the present setting because dislocations cannot be removed in this way. Indeed they act as a sort of obstacle when touching the profile, and this is overcome in Theorem~\ref{sollevamento}, where it is shown that a delicate truncation construction is still possible without affecting the dislocations. 

In Theorem \ref{prop:exnuova} we provide analytical support to the experimental evidence that, after nucleation, dislocations lie at the bottom. 
\begin{theorem}\label{th:22000}
	Assume $\B\neq\emptyset$, $d>2r_0 \ell$. Then there exist $\bar e>0$ and $\bar \gamma>0$ such that whenever $|e_0|>\bar e$, $\gamma> \bar \gamma$, and $e_0 (\bu_j\cdot\mathbf{e}_1)>0$ for all $\bu_j\in \mathbf{B}$, then any minimizer $(\bar{h}, \bar{\sigma}, \bar{H})$
	of the problem \eqref{minG1}  has all dislocations lying at the bottom of $\Om_h$, in the sense that the centers $z_i$ are of the form $z_i=(x_i, r_0)$.
\end{theorem}	

In the last part of the paper we study the nucleation of dislocations and we investigate conditions under which it is energetically favorable to create dislocations.  To this purpose, we modify  the energy  \eqref{total} by adding a term that accounts for the energy dissipated to create dislocations. Following the physical literature (see for instance \cite{N}), we assume that the energy cost of a new dislocation is proportional to the square of the norm of the corresponding Burgers vector. This leads to an energetic contribution $N(\sigma)$, given in \eqref{ns}. Therefore, our new  variational problem is to
\begin{equation}\label{vp}
\textrm{minimize}\quad F(h, \s, H)+N(\s)
\end{equation}
 among all admissible configurations $(h, \s, H)$, under a volume constraint, but without fixing the number of dislocations nor the Burgers vectors, which are allowed to be any integer multiple of certain fundamental directions in a set $\mathcal{B}^o\subset\mathbb{R}^2$.

 The regularity results of Section~\ref{sec:nonnucl} apply to the minimizers of \eqref{vp}. On the other hand, local and global minimizers of the minimum problem studied in Section~\ref{sec:nonnucl} may be regarded as local minimizers of \eqref{vp}.  Finally, in Theorem~\ref{th:ul} we identify a range of parameters for which all global minimizers have nontrivial dislocation measures (see \cite{LPM} for an analogous result in  heterogeneous nanowires). 
 
 \begin{theorem}\label{th:22001}
	Assume that there exists  $\bu\in \mathcal{B}^o$ such that $ \bu\cdot \mathbf{e}_1\neq 0$, and let $d>2r_0 \ell$. Then there exists $\bar \gamma>0$ such that whenever $|e_0|>\bar e$, and $\gamma> \bar \gamma$, where $\bar e$ is as in Theorem \ref{th:22000}, any minimizer $(\bar{h}, \bar{\sigma}, \bar{H})$
	of the problem \eqref{vp}  has nontrivial dislocations, i.e., $\bar{\sigma}\neq 0$.
	\end{theorem}

\section{Epitaxial elastic films with dislocations}\label{sec:nonnucl}
\subsection{Setting of the Problem}

We assume that the substrate is rigid and occupies the semi-infinite strip $(0,\ell)\times (-\infty, 0)$, and that the reference configuration of the elastic film is given by
$$
\Omega_h:=\left\{  z=\left(  x,y\right)  :\,0\leq x<\ell,\,0<y<h\left(
x\right)  \right\}  \label{omega}%
$$
with $h:\left[  0,\ell\right]  \rightarrow\left[  0,\infty\right)$. 
The
graph of $h$ represents the \emph{free} profile of the film and
the line $y=0$ corresponds to the film/substrate interface. 
The space of admissible profiles is defined by
$$
AP(0,\ell):=\left\{h:\R\to  [0,+\infty):\, \text{$h$ is lower  semicontinuous and $\ell$-periodic,}\, \operatorname*{Var}(h;0,\ell)<+\infty \right\}\,.
$$
Here $\operatorname*{Var}(h;0,\ell)$ denotes the {\em pointwise total variation} of $h$ over the interval
$(0,\ell)$, given by 
$$
 \operatorname*{Var}(h;0,\ell):=  \sup \sum_{i=1}^k |h(x_i)-h(x_{i-1})|<+\infty\,,
$$
where the supremum is taken over all partitions
$\{x_0,x_1,\dots,x_k\}$, with $0<x_0<x_1<\dots<x_k<\ell$, $k\in\N$.
Since $h\in AP(0,\ell)$ is $\ell$-periodic,  its pointwise total variation is finite over any bounded interval of
$\R$. Therefore,  it admits  right and left limits at every $x\in\R$
denoted by $h(x+)$ and $h(x-)$, respectively. 
In what follows we use the notation
\begin{equation}\label{piuomeno}
h^+(x):=\max\{h(x+), h(x-)\}\,, \qquad h^-(x):=\min\{h(x+), h(x-)\}\,.
\end{equation}
We set
$$
\Om^\#_h:=\{(x,y):x\in\R,\,0<y<h(x)\}
$$
to be the open set obtained by repeating copies of $\Om_h$ $\ell$-periodically in the $x$-direction.
We define 
$$
\Gamma_h:=\{(x,y):x\in[0,\ell),\,h^-(x)\leq y\leq h^+(x)\}\,,
$$
 and {\em the set of vertical cracks }
 \begin{equation}\label{Sigma_g}
 \Sigma_h:=\{(x,y):x\in[0,\ell)\,,\,h(x)<h^-(x),\,h(x)\leq y\leq h^-(x)\}\,.
\end{equation}
We also set 
$$
\widetilde\Gamma_h:=\Gamma_h\cup\Sigma_h\,,
$$
and we will use the notation
$$
\Gamma_h^\#:=\{(x,y)\in \R^2:\, x\in\R,\,h^-(x)\leq y\leq h^+(x)\}\,.
$$
Similarly we define $ \Sigma^\#_h$ and $\widetilde\Gamma_h^\#$. 

Observe that if $h\in AP(0,\ell)$, then
\begin{equation}\label{bound h}
\Vert h\Vert_\infty\le\frac1\ell\int_0^\ell h\,dx+\operatorname*{Var}(h;0,\ell)\le\frac{|\Omega_h|}{\ell}+\mathcal{H}^1(\Gamma_h)\,.
\end{equation}

We work within the
theory of small elastic deformations, so that
\[
{E}( {u})  :=\frac{1}{2}\left(  \nabla{u}+\nabla{u}^{T}\right)
\]
represents the strain, with ${u}:\Omega_h\rightarrow\mathbb{R}^{2}$ the
planar displacement. The elastic energy density is
 \beq\label{canonico}
W(E):=\frac{1}{2}\C E:E=\mu|E|^2+\frac{\lambda}{2}\bigl[{\rm tr}(E)\bigr]^2\,,
\eeq
where 
\begin{equation}\label{cxi}
\C E=
\left(
\begin{array}{cc}
(2\mu+\lambda)E_{11}+\lambda E_{22} & 2\mu E_{12} \\
2\mu E_{12} & (2\mu+\lambda)E_{22}+\lambda E_{11} \\
\end{array}
\right)
\end{equation}
and the {\em Lam\'e coefficients} $\mu$ and $\lambda$ satisfy the ellipticity conditions
\begin{equation}\label{lame}
\mu>0\qquad\text{and}\qquad\mu+\lambda>0\,.
\end{equation}

Throughout this section we assume the presence of $k$ dislocations with given Burgers vectors $\mathbf{B}:=\{\bu_1, \dots, \bu_k\}\subset\R^2$ and centers $\{z_1, \dots, z_k\}\subset \Om_h$ such that $B_{r_0}(z_i)\subset\Om_h^\#$, with $r_0\in (0, \ell/2)$ a (small) positive constant representing the {\em core radius} of the dislocations. 
With any such collection of dislocations we associate the {\it $\ell$-periodic dislocation measure}
$$
\sigma:=\sum_{i=1}^{k}\bu_i\delta^\#_{z_i}\,,
$$
where, given $z\in \Om_h$ we denote by $\de_{z}^\#$ the $\ell$-periodic extension of the Dirac delta $\de_{z}$, i.e.,  
$$
\de_{z}^\#:=\sum_{k\in\Z}\de_{z+k\ell\mathbf{e}_1}\,.
$$
To regularize $\sigma$, we fix a nonnegative radially symmetric $\ro \in C^{\infty}_c(B_1(0))$,  with $\int_{\R^2}\ro\, dz=1$, and we define
\begin{equation}
\ro_{r_0}(z):=\frac{1}{r_0^2}\ro\Bigl(\frac{z}{r_0}\Bigr)\quad\text{and}\quad\ro_{r_0}^\#:=\ro_{r_0}*\de^\#_0\,.
\label{mollifier}
\end{equation}
Note that $\ro_{r_0}^\#$ is  the $\ell$-periodic extension in the $x$-direction of the function $\ro_{r_0}$.

Given $h\in AP(0,\ell)$ we denote by $\md(\Om_h; \mathbf{B})$ the subset of the space of vector valued Radon measures $\mathcal{M}(\Om_h^\#; \R^2)$ defined by
\begin{multline*}
\md(\Om_h; \mathbf{B}):=\biggl\{\sigma\in \mathcal{M}(\Om_h^\#; \R^2):\, \sigma=\sum_{i=1}^{k}\bu_i\delta^\#_{z_i}, \, z_i\in \Om_h, \text{ with } B_{r_0}(z_i)\subset\Om_h^\#\biggr\}\,.
\end{multline*}
Observe that we are not requiring that the centers of the $k$ dislocations are all distinct, thus allowing for superpositions of different dislocations.

We recall that the curl of  a function $H$ with values in $\mathbb{M}^{2\times2}$ is defined by
$$
\curl H:=\Bigl(\frac{\partial H_{12}}{\partial x}-\frac{\partial H_{11}}{\partial y}, \frac{\partial H_{22}}{\partial x}-\frac{\partial H_{21}}{\partial y}\Bigr)\,.
$$

The total energy functional will depend on the film profile $h$ and on the dislocation measure $\sigma\in \md(\Om_h)$ via the associated {\em strain field} $H$ satisfying the constraint $\curl H=\sigma*\ro_{r_0}$, which accounts also for the interactions between the different dislocations. Moreover, the presence of a mismatch between the film and the substrate lattices is modeled by enforcing a Dirichlet boundary condition at the interface $\{y=0\}$, namely by requiring that the tangential trace of $H$ on the interface equals 
$e_0 \mathbf{e}_1$, where $\mathbf{e}_1:=(1,0)$ and $e_0\neq 0$.
 To be precise, we introduce the following set of admissible triples
\begin{multline}\label{X space}
X(e_0; \mathbf{B}):=\Big\{(h, \sigma, H):\, h\in AP(0,\ell),\, \sigma\in \md(\Om_h; \mathbf{B}),\, H\in \mathbf{H}_\#(\curl; \Om_h; \mathbb{M}^{2\times2})\\
\text{ such that } \curl H=\sigma*\ro_{r_0} \text{ in }\Omega_h\text{ and } H[\,\mathbf{ e_1}\, ]=e_0\mathbf {e_1}
\text{ on }\{y=0\}\Big\}\,,
\end{multline}
where we are using the fact that admissible fields $H$ admit a tangential trace (see, e.g., Chapter 4 in \cite{BF}), and where, denoting by $H^\#$ the $\ell$-periodic extension in the $x$-direction of $H$,
\begin{multline}\label{Hcurl}
 \mathbf{H}_\#(\curl; \Om_h; \mathbb{M}^{2\times2}):=\\
 \{H\in L^2_{loc}(\Om_h; \mathbb{M}^{2\times2}):\,  \curl H\in L^2(\Om_h; \R^2)\text{ and }\curl H^\#\in L^2_{loc}(\Om^\#_h;\R^2)\}\,.
\end{multline}
The total energy of the system is given by
\begin{equation}
F(h, \sigma, H)  :=\int_{\Omega_h}W(\sym{H})\, dz+\gamma\H^1(\Gamma_h)+2\gamma\H^1(\Sigma_h) \label{sharp model energy}%
\end{equation}
for every admissible  configuration $(h, \sigma, H)\in X(e_0; \mathbf{B})$,
where we recall that $
 \sym{H}:=(H+H^T)/2
 $
 and  $\gamma$ is a positive constant depending on the material properties.

For every fixed  profile $ h\in AP(0,\ell)$ and  dislocation measure $\sigma$ we denote by $H_{h, \s}$ the unique strain field that minimizes 
$$
H\mapsto \int_{\Om_h}W(H_{sym})\, dz
$$
over all $ H\in\Hc{h} $ such that $(h, \sigma, H)\in X(e_0; \mathbf{B})$. The existence and uniqueness of $H_{h, \s}$ follow from the coercivity and strict convexity of the energy \eqref{sharp model energy}  (see \eqref{cxi} and \eqref{lame}) and the fact that the Dirichlet condition in \eqref{X space} is preserved under weak convergence in 
the space $\mathbf{H}_\#(\curl; \Om_h; \mathbb{M}^{2\times2})$ (see \eqref{Hcurl}). Note that $H_{h, \s}$ is determined as the unique solution in $\Hc{h}$ to the  system 
\beq\label{ELHhs}
\begin{cases}
\curl H_{h,\s}=\sigma*\ro_{r_0} & \text{in $\Om_h$,}\\
\Div\, \C (H_{h,\s})_{sym}=0& \text{in $\Om_h$,}\\
\C (H_{h,\s})_{sym}[\nu]=0 &\text{on $\Gamma_h$,}\\
H_{h,\s}[\, \mathbf{e}_1\, ]=e_0\mathbf{e}_1 & \text{on $\{y=0\}$.}
\end{cases}
\eeq
Note also that if  $(h, \sigma, H_{h, \s})\in X(e_0; \B)$ is a (locally) minimizing configuration, with $h\in C^2_\#([0, \ell])$ and $h>0$, then by considering smooth variations of $h$ supported in the complement of the projection of $\cup_{i=1}^k \bar B_{r_0}(z_i)$ on the $[0, \ell]$, we obtain by standard arguments the following Euler-Lagrange equation
\beq\label{eq:EL}
\kappa+W((H_{h, \s})_{sym})=\Lambda \qquad\text{on }\overline{\Gamma_h\setminus \cup_{i=1}^k\bar B_{r_0}(z_i)}\,,
\eeq
where 
$$
\kappa:=-\biggl(\frac{h'}{\sqrt{1+h'^2}}\biggr)'
$$
denotes the curvature of $\Gamma_h$ and $\Lambda$ is the constant Lagrange multiplier associated with the volume constraint. 
This motivates the following definition.
\begin{definition}\label{def:critical}
Let $(h, \sigma, H_{h, \s})\in X(e_0; \B)$, with  $h\in C^2_\#([0, \ell])$ and $h>0$. We say that $(h, \sigma, H_{h, \s})$ is a {\em critical configuration}
if \eqref{ELHhs} and \eqref{eq:EL} are satisfied.
\end{definition}

In the sequel we will use the following {\em canonical decomposition} of $H_{h,\s}$:
$$
H_{h,\s}=e_0Du_h+K_{h,\s}\,,
$$
where $u_h$ is the elastic equilibrium in $\Om_h$ such that $u_h(x,0)=(x,0)$, that is the unique solution
to the system
\beq\label{eleq}
\begin{cases}
\Div\, \C E(u_h)=0 & \text{in $\Om_h$,}\\
 \C E(u_h)[\, \nu\, ]=0 & \text{on $\Gamma_h$,}\\
 u_h(x,0)=(x,0) & \text{on $\{y=0\}$,} 
\end{cases}
\eeq
such that $(x,y)\in\Omega^\#\mapsto u_h(x,y)-(x,0)$ belongs to
\begin{multline*}
LD_\#( \Om_h;\R^2){:=} \left\{v\in L^2_{\rm loc}(\Om^\#_h;\R^2):\, v(x,y)=v(x{+}\ell,y)\right. \\
\left. \text{ for }(x,y)\in \Om_h^\#\,, E(v)|_{\Om_h}\in L^2(\Om_h;\R^2)\right\}\,,
\end{multline*}
and $K_{h,\s}$ is the unique solution in $\Hc{h}$ to 
\beq\label{Kcanonical}
\begin{cases}
\curl  K_{h, \s}= \sigma*\ro_{r_0} & \text{in $\Om_h$,}\\
\Div\, \C(K_{h,\s})_{sym}=0& \text{in $\Om_h$,}\\
\C(K_{h,\s})_{sym}[\nu]=0 &\text{on $\Gamma_h$,}\\
 K_{h, \s}[\, \mathbf{e}_1\,]=0 & \text{on $\{y=0\}$.}
\end{cases}
\eeq
We set
\beq\label{v0}
v_0(x,y):=\Bigl(x, \frac{-\lambda y}{2\mu+\lambda}\Bigr)\quad\text{and}\quad W_0:=W(E(v_0))\,.
\eeq 
Observe that $v_0$ is the elastic equilibrium corresponding to the flat configuration and $e_0=1$.

\subsection{Existence}\label{subsec:existence}

We start with the following Korn-type inequality.

\begin{lemma}\label{lm:korn-fusco}
Let $\Om\subset\R^2$ be a bounded open simply connected set with Lipschitz boundary and let $\Gamma$ be a non-empty connected  relatively open subset of $\pa\Om$.
Then, 
there exists a constant $C>0$ depending only on $\Omega$ and $\Gamma$ such that
\begin{equation}\label{estimate H}
\|H\|_{L^2(\Om; \mathbb{M}^{2\times2})}\leq C\big(\|H_{sym}\|_{L^2(\Om; \mathbb{M}^{2\times2})}+\|\curl H\|_{L^2(\Om; \R^2)}\big)
\end{equation}
for all $H\in  \mathbf{H}(\curl; \Om; \mathbb{M}^{2\times2})$ with tangential trace $H[\tau]=0$ on $\Gamma$. 
\end{lemma}
\begin{proof} \noindent {\bf Step 1.} We start by assuming that $\mathcal H^{1}(\pa\Om\setminus \Gamma)>0$ and,  without loss of generality, that 
$H_{sym}\in L^2(\Om; \mathbb{M}^{2\times2})$. Let
$$
K:=
\left(
\begin{array}{cc}
-D_yw_1 & D_xw_1\\
-D_yw_2 & D_xw_2
\end{array}
\right),
$$
where $w=(w_1, w_2)$ is the unique  solution to
$$
\begin{cases}
\Delta w= \curl H & \text{in $\Om$,}\\
w=0 & \text{on $\pa\Om\setminus\Gamma$,}\\
D_\nu w=0 & \text{on $\Gamma$.}\
\end{cases}
$$
By multiplying $\Delta w_i=(\curl H)_i$ by $w_i$, $i=1,2$ and integrating by parts, it follows from the Poincar\'e inequality
\begin{equation}\label{ineq1}
\|K\|_{L^2(\Om; \mathbb{M}^{2\times2})}= \|Dw\|_{L^2(\Om; \mathbb{M}^{2\times2})}\leq C \|\curl H\|_{L^2(\Om; \R^2)}\,.
\end{equation}
Since $\curl(H-K)=0$ in $\Omega$, by the Helmholtz decomposition theorem (see, e.g., Theorem 3.3.7 in \cite{M}) there exists $u\in H^1(\Omega;\mathbb{R}^2)$ such that $Du=H-K$. Moreover, $u$ is unique up to a constant. Since $(H-K)[\tau]=0$ on $\Gamma$, we can take $u=0$ on $\Gamma$. 
Using Korn's inequality (see, e.g., \cite{Ni}), we have
\begin{align}
\|Du\|_{L^2(\Om; \mathbb{M}^{2\times2})}\leq C \|E(u)\|_{L^2(\Om; \mathbb{M}^{2\times2})}&=
 C \|H_{sym}-K_{sym}\|_{L^2(\Om; \mathbb{M}^{2\times2})}\nonumber\\
 &\leq C\big(\|H_{sym}\|_{L^2(\Om; \mathbb{M}^{2\times2})}+ \|\curl H\|_{L^2(\Om; \R^2)}\big)\,,\label{ineq2}
\end{align}
where in the last inequality we have used  \eqref{ineq1}.
By \eqref{ineq1} and \eqref{ineq2}, we obtain \eqref{estimate H}.

\noindent {\bf Step 2.} If $\mathcal H^{1}(\pa\Om\setminus \Gamma)=0$, then the argument is similar, and it suffices to replace the condition $w=0$ on $\pa\Om\setminus \Gamma$ by $\int_\Om w\, dz=0$.
\end{proof}

The next lemma provides a useful elliptic estimate for the solutions to  systems of the type  \eqref{Kcanonical}.
\begin{lemma}\label{lm:ellcan}
Let $h\in AP(0,\ell)\cap Lip(0,\ell)$, $h\geq c_0>0$, $\|h'\|_{\infty}\leq M$ and let  $f\in L^2(0,\ell; \R^2)$. Then, there exists a constant $C>0$, depending only on
$c_0$ and $M$, such that if $H\in \Hc{h}$ is the solution to
$$
\begin{cases}
\curl  H= f & \text{in $\Om_h$,}\\
\Div\, \C H_{sym}=0& \text{in $\Om_h$,}\\
\C H_{sym}[\nu]=0 &\text{on $\Gamma_h$,}\\
 H [\, \mathbf{e}_1\,]=0 & \text{on $\{y=0\}$,}
\end{cases}
$$
then
\beq\label{ellcan1}
\|H\|_{L^2(\Om_h; \mathbb{M}^{2\times 2})}\leq C\|f\|_{L^2(0,\ell; \R^2)}\,.
\eeq
\end{lemma}
\begin{proof} Since $h\ge c_0>0$, the set $\Omega_h$ is connected, and since its complement is also connected, we have that $\Omega_h$ is simply connected. 
Hence, we can argue as in the proof of Lemma~\ref{lm:korn-fusco} to split $H=Du+K$, where $K$ is defined
$$
K=
\left(
\begin{array}{cc}
-D_yw_1 & D_xw_1\\
-D_yw_2 & D_xw_2
\end{array}
\right)
$$
with $w=(w_1, w_2)$ the unique  solution to
$$
\begin{cases}
\Delta w= f & \text{in $\Om_h$,}\\
w=0 & \text{on $\Gamma_h$,}\\
D_\nu w=0 & \text{on $\{y=0\}$.}\
\end{cases}
$$
As before we have that $\|K\|_{L^2(\Om_h;\mathbb{M}^{2\times2})}\leq C\|f\|_{L^2(\Om_h;\R^2)}$.
Note that $u\in H^1_\#(\Om_h;\R^2)$ can be  chosen to be identically $0$ on $\{y=0\}$ and solves
$$
\begin{cases}
\Div\, \C E(u)=-\Div\, \C K_{sym} & \text{in $\Om_h$,}\\
\C E(u)[\nu]=- \C K_{sym}[\nu] & \text{on $\Gamma_h$,}\\
u=0 & \text{on $\{y=0\}$.}
\end{cases}
$$
 Multiplying both sides of the equation above  by $u$, integrating by parts, and using the fact that if $H\in \mathbb{M}^{2\times2}$ is symmetric, then so is $\C H$ (see \eqref{cxi}),  we get
$$
\int_{\Om_h}\C E(u): E(u)\, dz=-\int_{\Om_h}\C K_{sym}: E(u)\, dz\,.
$$
Hence, also by Korn's inequality, we have 
$$
\|Du\|_{L^2(\Om_h;\mathbb{M}^{2\times2})}\leq C \|K\|_{L^2(\Om_h;\mathbb{M}^{2\times2})}\leq  C\|f\|_{L^2(\Om_h;\R^2)}\,,
$$
and we conclude that \eqref{ellcan1} holds.
\end{proof}

\begin{theorem}\label{th:existence}
The minimization problem 
\beq\label{minG}
\min\{F(h,\sigma, H):\, (h, \sigma, H)\in X(e_0; \mathbf{B}), \, |\Om_h|=d\}\,.
\eeq
admits a solution.
\end{theorem}
\begin{proof}
Let $\{(h_n, \sigma_n, H_n)\}\subset X(e_0; \mathbf{B})$ be a minimizing sequence. By the compactness results in \cite[Proposition~2.2 and Lemma~2.5]{FFLM}, we may assume that, up to a  subsequence (not relabeled), there exists $h\in AP(0,\ell)$ such that 
\begin{itemize}
\item[i)] $h_n\to h$ in $L^1(0,\ell)$;
\item[ii)] $\R^2\setminus\Om^\#_{h_n}\to \R^2\setminus\Om^\#_h$ in the sense of the Hausdorff metric.
\end{itemize}
Moreover,  in \cite[Lemma~2.1]{BC} it is shown that 
\begin{equation}\label{semi1}
\H^1(\Gamma_{h})+2\H^1(\Sigma_h)\leq \liminf_n\bigl[\H^1(\Gamma_{h_n})+2\H^1(\Sigma_{h_n})\bigr] \,.
\end{equation}
Setting $\sigma_n=\sum_{i=1}^{k}\bu_i\delta^\#_{z_{i,n}}$,  we can assume (up to extracting a further subsequence if needed) that 
 $z_{i,n}\to z_i\in \overline{\Om_h}$, with $B_{r_0}(z_i)\subset \Om^\#_h$. Note that if $z_i\cdot \mathbf{e}_1=\ell$ using the lateral periodicity we can assume that $z_i\cdot \mathbf{e}_1=0$, and so by \eqref{omega}  we have that $z_i\in \Om_h$.

Set $V_n:=\Omega_{h_n}\cup ((0,\ell)\times(-1,0])$ and $V:=\Omega_{h}\cup ((0,\ell)\times(-1,0])$. Since $H_n[\,\mathbf{ e_1}\, ]=e_0\mathbf {e_1}$ 	on $\{y=0\}$, by setting $H_n:=\nabla u_0$ in $(0,\ell)\times(-1,0]$, where $u_0(x,y):=(e_0x,0)$, we have that $H_n\in \mathbf{H}(\curl;  V_n;\mathbb{M}^{2\times2})$. Note that the sets $V_n$ are simply connected. Consider an increasing sequence of simply connected Lipschitz sets $U_j\subset V$ such that $(0,\ell)\times(-1,0]\subset  U_j$, $\pa U_j\cap \Gamma_h=\emptyset$ and $\cup_{j\in\N}U_j=V$. By Lemma~\ref{lm:korn-fusco} we have that for every $j$, the strain fields $H_n$ are equibounded in 
	 $L^2(U_j; \M)$.
Note also  that $\curl H_n=\sigma_n*\ro_{r_0}\to \sigma*\ro_{r_0}$ in $L^2(V; \R^2)$, where 
 $\sigma:=\sum_{i=1}^{k}\bu_i\delta^\#_{z_{i}}$. Thus,  by a diagonalization  argument, we may find $H\in \mathbf{H}(\curl;  V;\mathbb{M}^{2\times2})$ such that 
 $\curl H= \sigma*\ro_{r_0}$, and, up to the extraction of a further subsequence (not relabeled), $H_n\wto H$ weakly in $L^2(U_j; \M)$ for every $j$. Since $H_n=\nabla u_0$ in $(0,\ell)\times(-1,0]$, we have that $H=\nabla u_0$ in $(0,\ell)\times(-1,0]$, and, in turn,  $H[\mathbf{e}_1]=e_0\mathbf{e}_1$ on $\{y=0 \}\cap \pa\Om_h$. 
 It follows that $(h, \sigma, H)\in X(e_0; \mathbf{B})$ and for every $j\in \N$
 \begin{equation}\label{semi2}
 \int_{U_j\cap \Omega_h}W(H_{sym})\, dz\leq \liminf_n\int_{U_j\cap \Omega_h}W((H_n)_{sym})\, dz\leq  \liminf_n\int_{\Om_{h_n}}W((H_n)_{sym})\, dz.
 \end{equation}
 By  \eqref{semi1} and \eqref{semi2} and the  arbitrariness of $j$ we conclude that
 $$
 F(h, \sigma, H)\leq\liminf_n F(h_n, \sigma_n, H_n)\,.
 $$ 
 Thus $(h, \sigma, H)$ is a global minimizer. 
 \end{proof}

\subsection{Regularity}
In this subsection we establish the regularity properties of minimizers of problem \eqref{minG}. We shall follow the general strategy developed in \cite{FFLM,FuMo} to which we refer for all parts of the proofs that will remain unchanged.
\begin{theorem}\label{sollevamento} Let $d\geq 2r_0\ell$ and let $(\bar{h},\bar{\sigma}, {H_{\bar h, \bar\s}})$ be a minimizing configuration
for problem \eqref{minG} such that $\bar h^-$ is not flat. There exist $\beta>0$ depending only on 
$\|\bar h-d/\ell\|_{L^2(0,\ell)}$ and $F(\bar h, \bar\s, H_{\bar h, \s})$, and  $\Lambda>0$ depending on $\mu,\,\lambda,\,e_0,\,r_0$ and  $\beta$,
  such that  $(\bar{h},\bar{\sigma}, {H_{\bar h, \bar\s}})$ is also a minimizer of
\beq\label{minGpen}
\min\biggl\{F(h,\sigma,H)+\beta\int_0^\ell|h-\bar h|^2\, dx+\Lambda\bigl||\Om_h|-d\bigr|:\,(h,\sigma,H)\in X(e_0;\mathbf{B})\biggr\}\,.
\eeq
\end{theorem}
Before giving the proof we need the following technical lemma.
\begin{lemma}\label{lm:claim}
For all $\e>0$ there exists $\Lambda(\e)$ (depending also  on $\beta$, $\mu$, $\lambda$, $e_0$, and $r_0$)  with the following property:  For all 
$\Lambda\geq \Lambda(\e)$ if $(g, \tau, H_{g, \tau})$ is a minimizer of \eqref{minGpen}, with  $|\Om_g|>d$,  
$\tau=\sum_{i=1}^k\bu_i\delta^\#_{z_{i}}$, and if
$\Gamma'\subset\pa B_{r_0}(z_j)\cap \Gamma_g$ for some $j\in\{1,\dots, k\}$, with $z_j\cdot e_2>r_0$,   is any  connected arc, then $\H^1(\Gamma')\leq \e$.
\end{lemma}
\begin{proof}
In order to prove the lemma observe that in $B_{r_0}(z_j)$ we can write $H_{g,\tau}= Dv+K$, where 
 $$
K:=
\left(
\begin{array}{cc}
k_1 & 0\\
k_2 & 0
\end{array}
\right)\,,
$$
with
 
$$
k_l(x,y):=-\sum_{i=1}^k(\bu_i\cdot \mathbf{e}_l)\int_0^y\ro_{r_0}^\#(x-x_i, t-y_i)\, dt \text{ for $l=1,2$,}
$$
where $\ro_{r_0}^\#$ is defined in \eqref{mollifier}, 
and $v\in H^1_\#(\Om_h;\R^2)$ satisfies
$$
\begin{cases}
\Div\, \C E(v)=-\Div\, \C K_{sym} & \text{in $\Om_g$,}\\
\C E(v)[\nu]=- \C K_{sym}[\nu] & \text{on $\Gamma_g$,}\\
v=0& \text{on $\{y=0\}$.}
\end{cases}
$$
Since $K$ and $\Gamma'$ are both smooth, $v$ is smooth in $B_{r_0}(z_j)\cup \Gamma'$. Let $\Gamma''\subset\Gamma'$ be the subarc with the center of $\Gamma'$ and such that $\H^1(\Gamma'')=\frac12\H^1(\Gamma')$. By  elliptic estimates for the Lam\'e system (see for instance \cite[Proposition~8.9]{FuMo}) there exists a constant
$C_1>0$ depending only on $\H^1(\Gamma')$,  $r_0$, the Lam\'e coefficients $\mu$ and $\lambda$, and on $F(\bar h, \bar \s, H_{\bar h, \bar \s})$, such that 
$$
\sup_{\Gamma''}|Dv|\leq C_1\,.
$$
In particular, the constant $C_1=C_1(\Gamma')$ above is uniformly bounded if $\H^1(\Gamma')$ is bounded away from $0$.  In turn, we  obtain
\beq\label{eq:sol1}
\sup_{\Gamma''}|H_{g,\tau}|\leq C_1+C_2\,,
\eeq
where the constant $C_2>0$ depends only on $r_0$.  

Fix  $\vphi\in C^{\infty}_c(I)$, $\vphi\geq 0$,  where $I$ is an open interval contained in the projection of $\Gamma''$  onto the $x$-axis. Since $z_j\cdot e_2>r_0$, for $t>0$ sufficiently small we have that $B_{r_0}(z_j-t\Vert\vphi\Vert_\infty)\subset \Omega_{g-t\vphi}^\#$ and so we can take as admissible competitor the triple $(g-t\vphi, \s_t, H_{t})$, where 
$\tau_t:=\sum_{i\neq j}\bu_i \delta^\#_{z_i}+\bu_j\delta^\#_{z_{j}-t\|\vphi\|_{\infty}\mathbf{e}_2}$, 
$H_t:=H_{g, \tau}+K_t$, where
$$
K_t:=
\left(
\begin{array}{cc}
k_{t,1} & 0\\
k_{t,2} & 0
\end{array}
\right)\,,
$$
with 
$$
k_{t,l}(x,y):=-\biggl(\int_0^y\ro_{r_0}^\#(x-x_j,s-y_j-t\|\vphi\|_{\infty})\, ds-\int_0^y\ro_{r_0}^\#(x-x_j,s-y_j)\, ds\biggr)\bu_j\cdot \mathbf{e}_l\,, \text{ for $l=1,2$.}
$$
By minimality, we have
$$
F(g-t\vphi, \tau_t, H_{t})+\beta\int_0^\ell|g-t\vphi-\bar h|^2\, dx+\Lambda(|\Om_{g-t\vphi}|-d)\geq 
F(g, \tau, H_{g, \tau})+\beta\int_0^\ell|g-\bar h|^2\, dx+\Lambda(|\Om_{g}|-d)\,.
$$
By dividing both sides by $t>0$ and letting $t\to 0^+$, we obtain
\begin{multline}\label{100}
\int_{\Om_g}\C ((H_{g,\tau})_{sym}):\dot{K}_{sym}\, dz+\int_I W((H_{g,\tau})_{sym})(x, g(x))\vphi(x)\, dx\\
-\gamma\int_I\frac{g'\vphi'}{\sqrt{1+g'^2}}\, dx-2\beta\int_I(g-\bar h)\vphi\, dx-\Lambda\int_I\vphi\, dx\geq 0\,,
\end{multline}
where 
$$
\dot{K}_{sym}:=
\left(
\begin{array}{cc}
\dot{k}_{1} & {\dot{k}_{2}}/2\\
{\dot{k}_{2}}/2 & 0
\end{array}
\right),
\qquad \dot{k}_{l}(x,y):=\|\vphi\|_{\infty} \ro_{r_0}(z-z_j)\bu_j\cdot \mathbf{e}_l \text{ for $l=1,2$.}
$$
Since $\Gamma''\subset\pa B_{r_0}(z_j)\cap \Gamma_g$, integrating by parts we get
$$-\gamma\int_I\frac{g'\vphi'}{\sqrt{1+g'^2}}\, dx
\le \frac{\gamma}{r_0}\Vert\varphi\Vert_{\infty}\ell\,.$$
Thus, by taking a sequence $\{\vphi_n\}$ as above  converging pointwise to $1$ in $I$, from \eqref{100} we get that  there exists $C_3>0$ depending only on $r_0$ and the Lam\'e coefficients $\lambda$, $\mu$, such that 
\beq\label{eq:claim1}
\Lambda \H^1(\Gamma')\leq c(r_0)\Lambda \mathcal{L}^1(I) \leq C_3\biggl(\int_{\Om_g}|(H_{g,\tau})_{sym}|\, dz+\ell\sup_{\Gamma''}|H_{g,\tau}|^2 +\frac{\gamma\ell}{r_0}+\beta\int_I|g-\bar h|\, dx\biggr)\,,
\eeq
where we used the fact that $\H^1(\Gamma'')=\frac12\H^1(\Gamma')$.
Now assume by contradiction that there exist $\Lambda_n\to+\infty$ and minimizers 
$(g_n, \tau_n, H_{g_n, \tau_n})$ of \eqref{minGpen}, with $|\Om_{g_n}|>d$, $\tau_n=\sum_{i=1}^k\bu_i\delta^\#_{z_{i, n}}$, and $\Gamma'_n\subset \Gamma_{g_n}\cap \pa B_{r_0}(z_{j,n})$ for some $j\in \{1, \dots, k\}$, with 
$$
\inf_n \H^1(\Gamma'_n)>0\,.
$$
Thus, from \eqref{eq:sol1} we deduce that 
$$
\sup_{\Gamma''_n}|H_{g_n, \tau_n}|\leq C_4\,,
$$
with $C_4$ independent of $n$. 
Recalling \eqref{eq:claim1} and observing that by mininimality 
$$\sup_{n}\biggl(\|(H_{g_n, \tau_n})_{sym}\|_{L^2(\Om_{g_n}; \mathbb{M}^{2\times2})}+\beta\int_0^\ell|g_n-\bar h|^2\, dx\biggr)<+\infty\,,$$ we conclude that 
$$
\Lambda_n\H^1(\Gamma'_n)\leq C
$$
for some constant $C$ independent of $n$,  which is impossible since $\Lambda_n\to+\infty$. 
\end{proof}

\begin{proof}[Proof of Theorem~\ref{sollevamento}]
We fix $\beta$ such that 
\beq\label{eq:beta}
F(\bar h, \bar \s, H_{\bar h,\bar  \s})<\frac{\beta}4 \int_0^b\Bigl|\bar h-\frac{d}b\Bigr|^2\, dx\,.
\eeq
In order to prove the result we will show that any minimizing configuration $(g, \tau, H_{g, \tau} )$ for \eqref{minGpen} satisfies 
the volume constraint $|\Om_g|=d$, provided that $\Lambda$ is sufficiently large. 
We argue by contradiction and consider several cases.

\noindent {\bf Step 1.} 
 If $|\Om_{g}|<d$, then define $h:=g+(d-|\Om_{g}|)/\ell$ and for all $(x,y)\in\Om_{ h}$
$$
{ H}(x,y):=
\begin{cases}
\displaystyle e_0Dv_0(x,y) & \text{if $\displaystyle{0<y<\frac{d-|\Om_{g}|}{\ell}}$}\,, \vspace{4pt} \\
\displaystyle H_{g, \tau}\Bigl(x,y-\frac{d-|\Om_{g}|}{\ell}\Bigr) & \text{if $\displaystyle{y\geq\frac{d-|\Om_{g}|}{\ell}}$}\,,
\end{cases}
$$
where $v_0$ is defined as in \eqref{v0} and ${ \sigma}$ is the dislocation measure obtained by moving in the $e_2$ direction all the centers $z_i$, $i=1,\dots,k$ of $\tau$ by the vector $(d-|\Om_{g}|)\mathbf{e}_2/\ell$.
Then by \eqref{v0},
\begin{align*}
F({ h},{\sigma}, { H})&+\beta\int_0^\ell|h-\bar h|^2\,dx+ \displaystyle  \Lambda\bigl||\Om_{ h}|-d\bigr|-F(g,\tau, H_{g, \tau})-
\beta\int_0^\ell|g-\bar h|^2\, dx-\Lambda\bigl||\Om_{g}|-d\bigr|\\
 & \displaystyle = e_0^2W_0(d-|\Om_{g}|)+\beta\int_0^\ell\frac{d-|\Om_g|}{\ell}\Bigl(2(g-\bar h)+\frac{d-|\Om_g|}{\ell}\Bigr)\, dx-\Lambda(d-|\Om_{g}|)\\
 & \displaystyle \leq e_0^2W_0(d-|\Om_{g}|)-\Lambda(d-|\Om_{g}|)\,,
\end{align*}
where we used the fact that $\int_0^\ell g\,dx=|\Omega_g|<d
=\int_0^\ell \bar h\,dx$. By taking $\Lambda>e_0^2W_0$, we obtain  a contradiction to the minimality of $(g , \tau, H_{g, \tau})$.

\noindent{\bf Step 2.}  If $|\Om_{g}|>d$, we distinguish two cases. Let $y_{max}$ be the maximal height of points in $\Gamma_{g}$ and for all $i=1,\dots,k$ write $z_i=(x_i,y_i)$.

\noindent{\bf Case 1.} If $y_i<y_{max}-r_0$ for all $i=1,\dots,k$,  we truncate  $g$ in such a way that, denoting by $ h$ the resulting function, we still have  $B_{r_0}(z_i)\subset\Om^\#_{ h}$ for all $i$ and $|\Om_{ h}|\geq d$.  
Since $h\le g$,  we  can estimate
\begin{align*}
F( h, \tau, H_{g, \tau})&\displaystyle+\beta\int_0^\ell|h-\bar h|^2\, dx+\Lambda\bigl(|\Om_{ h}|-d\bigr)-F(g,\tau, H_{g, \tau})-\beta\int_0^\ell|g-\bar h|^2\, dx-\Lambda\bigl(|\Om_{g}|-d\bigr)\\
&\displaystyle\leq\beta\int_0^\ell(g-h)(2\bar h-h-g)\, dx-\Lambda\int_0^\ell(g-h)\, dx\\
& \displaystyle \leq \bigl(2\beta\|\bar h\|_\infty-\Lambda\bigr)\int_0^\ell(g-h)\, dx
\leq (2\beta C_0-\Lambda)\int_0^\ell(g-h)\, dx<0\,,
\end{align*}
provided $\Lambda>2\beta C_0$, which would contradict  the minimality of $({\bar h},{\bar\sigma}, {\bar H})$. Note that the constant   $C_0$  bounding $\|\bar h\|_\infty$ from above only depends on $F(\bar h, \bar\s, H_{\bar h, \bar\s})$ (see \eqref{bound h}). 

\noindent{\bf Case 2.} Assume now that there exists $j$ such that    $y_j=y_{max}-r_0$. 
We claim that  for every $i\in\{1,\ldots,k\}$ the intersection 
$\Gamma_g\cap \pa B_{r_0}(z_i)$ is either empty or a (possibly degenerate) connected  arc. Indeed, if this were not true for some $i\in\{1,\ldots,k\}$, we could find
two points $w_1$, $w_2\in \Gamma_g\cap \pa B_{r_0}(z_i)$ such that the graph of $g$ is detached from $\pa B_{r_0}(z_i)$ above the arc $\widehat{w_1w_2}$ connecting $w_1$ and $w_2$ on $\partial B_{r_0}(z_i)$. Denote by $D$ the region bounded by $\widehat{w_1w_2}$ and the  arc on $\Gamma_g$ connecting the two points.  Fix a point $w$ in the interior of 
$\widehat{w_1w_2}$ and consider the tangent to $\pa B_{r_0}(z_i)$ at $w$. Moving this tangent outward in the direction $w-z_i$, we cut out 
a region $D'\subset D$ bounded by this line and $\Gamma_g$ such that $|D'|\leq |\Om_g|-d$. Note that by doing so we get a new profile $\hat g$ such that $\H^1(\Gamma_{\hat g})<\H^1(\Gamma_g)$ and, in turn, 
\begin{equation}\label{case2}
F(\hat g, \tau, H_{g,\tau})<F(g, \tau, H_{g,\tau})\,.
\end{equation}
Therefore, arguing as in the previous step, we contradict the minimality of $(g, \tau, H_{g,\tau})$, provided that $\Lambda$ is chosen as before. Thus, the claim holds.

Set 
$$
	J:=\{j\in\{1,\ldots,k\}:\,y_j=y_{max}-r_0\}\,.
$$
Since $y_{max}\ell\ge|\Om_{g}|>d\geq 2r_0\ell$, we have that $y_{max}-2r_0=:\delta>0$.  Hence, $y_j=r_0+\delta$ for every $j\in J$. Let
\begin{equation}\label{epsilon}
0<\varepsilon<\min\{\delta,\ell\}/k\,. 
\end{equation} 
Let $\Lambda_\varepsilon>0$ be so large that
\begin{equation}\label{lambda large}
\frac{1}{\Lambda}F(\bar h, \bar\s, H_{\bar h, \bar\s})<\varepsilon
\end{equation}
for all $\Lambda>\Lambda_\varepsilon$.
Fix $j\in J$ and assume that $0<x_j<\ell$ (the cases $x_j=0$ and $x_j=\ell$ are similar). By the previous claim, the set $\Gamma_g\cap \pa B_{r_0}(z_j)$ is a (possibly degenerate) connected  arc $\Gamma_j$ of left endpoint $p_j$ and right endpoint $q_j$. 

Since $y_j\ge r_0+\delta$, we may  apply Lemma \ref{lm:claim} to conclude that, choosing a possibly larger $\Lambda_\varepsilon$, then $\mathcal{H}^1(\Gamma_j)<\varepsilon$.
Let $\Pi_2:\mathbb{R}^2\to\mathbb{R}$ be the projection onto the $y$-axis. Then 
$ \mathcal{L}^1(\Pi_2(\Gamma_j)\le  \mathcal{H}^1(\Gamma_j)<\varepsilon$.  Hence, 
\begin{equation}\label{qj}
q_j\cdot e_2\ge y_{max}-\varepsilon= 2r_0+\delta-\varepsilon\,.
\end{equation}
If $q_j$ belongs to $\Gamma_g\cap \pa B_{r_0}(z_{j_1})$ for some $j_1\ne j$, then by \eqref{epsilon} and \eqref{qj}, 
$$y_{j_1}=(z_{j_1}-q_j)\cdot e_2+q_j\cdot e_2\ge -r_0+2r_0+\delta-\varepsilon=r_0+\delta-\varepsilon$$  
Let $q_{j_1}$ be  the right endpoint of the (possibly degenerate) connected  arc $\Gamma_g\cap \pa B_{r_0}(z_{j_1})$. Since $y_{j_1}>r_0$ by Lemma \ref{lm:claim} and \eqref{qj} we obtain as before that the arc $\Gamma_{j_1}$ of endpoints $q_{j}$ and $q_{j_1}$ has length less than $\varepsilon$ and that $q_{j_1}\cdot e_2\ge 2r_0+\delta-2\varepsilon$. 
If $q_{j_1}$ belongs to $\Gamma_g\cap \pa B_{r_0}(z_{j_2})$ for some $j_2\ne j_1$, we continue this process, otherwise we stop and repeat a similar procedure for the left endpoint $p_j$. Let $J_j$ be the set of the indices $i\in\{1,\ldots,k\}$ corresponding to balls selected in this procedure. Note that by construction 
$y_i>r_0$ for every $i\in J_j$, and so
\begin{equation*}\label{vertical balls}
\sum_{j\in J}\sum_{i\in J_j} \mathcal{L}^1(\Pi_2(\Gamma_g\cap \pa B_{r_0}(z_i)))\le \sum_{j\in J}\sum_{i\in J_j} \mathcal{H}^1(\Gamma_g\cap \pa B_{r_0}(z_i))\le k\varepsilon\,.
\end{equation*}
Since the union of all the arcs $\Gamma_g\cap \pa B_{r_0}(z_i)$ is connected and $\Gamma_j$ is one of them this implies that 
\begin{equation}\label{graph above j}
y_{max}-k\varepsilon\le g(x) \le y_{max} 
\end{equation}
for all $x\in (0,\ell)$ such that $(x,g(x))\in \Gamma_g\cap \pa B_{r_0}(z_i)$ for some $i\in J_j$.\newline
Let $\Pi_1:\mathbb{R}^2\to\mathbb{R}$ be the projection onto the $x$-axis. Since 
$$\sum_{j\in J}\sum_{i\in J_j} \mathcal{L}^1(\Pi_1(\Gamma_g\cap \pa B_{r_0}(z_i)))\le \sum_{j\in J}\sum_{i\in J_j} \mathcal{H}^1(\Gamma_g\cap \pa B_{r_0}(z_i))\le k\varepsilon<\ell\,,$$ 
the open set $U:=(0,\ell)\setminus \cup_{j\in J}\cup_{i\in J_j} \Pi_1(\Gamma_g\cap \pa B_{r_0}(z_i))$ is nonempty.

\noindent{\bf Case 2a.} 
Assume that there exists a connected component $I_i$ of $U$ and $s<t\in I_i$  such that $\Gamma_g\cap (s,t)\times \R$ lies strictly above the segment $\gamma$ connecting 
$(s, g^-(s))$ with $(t, g^-(t))$. Let $\nu$ be the unit vector orthogonal to $\gamma$ and pointing upward.  Moving $\gamma$ in the direction of $\nu$, we can choose $\eta>0$ so that  the region $D$ bounded by the segment $\gamma+\eta \nu$ and $\Gamma_g\cap (s,t)\times \R$ satisfies $|D|\leq |\Om_g|-d$ and $D\cap \cup_{i=1}^kB_{r_0}(z_i)=\emptyset$. Then, arguing as in the proof of \eqref{case2} we get a contradiction provided that $\Lambda$ is chosen as before.

\noindent{\bf Case 2b.}
 For every connected component $I_i$ of the set $U$ we have that $g^-$ is a convex function in the interval $I_i$. In this case we claim that there exists a constant $c>0$ independent of $g$ such that
\begin{equation}\label{graph above}
y_{max}-c\varepsilon\le g(x) \le y_{max}\quad\text{for all }x\in (0,\ell)\,. 
\end{equation}
In view of \eqref{graph above j} it suffices to prove \eqref{graph above} in each $I_i$.
Fix $I_i$ and let $a_i$ be its left endpoint. Then the point $(a_i,g(a_i))$ belongs to one of the balls $B_{r_0}(z_l)$ for some $j\in J$ and $l\in J_j$. Let $\theta_i$ be the angle that the oriented segment of endpoints $z_l$ and $(a_i,g(a_i))$ forms with the $x$-axis. By \eqref{graph above j}, we have that $\theta_i\ge \frac{\pi}{4}$ for $\varepsilon$ sufficiently small.
 Since $g$  is a convex function in the interval $I_i$, it lies above the line 
 $$t\mapsto (a_i,g(a_i))+t\biggl(1,-\frac{\cos\theta_i}{\sin\theta_i}\biggr)$$ tangent to the ball 
 $\pa B_{r_0}(z_l)$ at $(a_i,g(a_i))$. Since $\mathcal{H}^1(\Gamma_g\cap \pa B_{r_0}(z_l))\le\varepsilon$, we have that $\cos\theta_i\le \cos (\pi/2-\varepsilon/r_0)=\sin (\varepsilon/r_0)\le \varepsilon/r_0$. Hence, for $t>0$, 
 $$g(a_i)-t\frac{\cos\theta_i}{\sin\theta_i}\ge g(a_i)-\frac{t\sqrt{2}}{r_0}\varepsilon\ge y_{max}-k\varepsilon-\frac{\ell\sqrt{2}}{r_0}\varepsilon\,,$$
where in the last inequality we used \eqref{graph above j}. This proves that \eqref{graph above} holds.
By \eqref{minGpen} we have 
$$
F(g,\tau,H_{g,\tau})+\beta\int_0^b|g-\bar h|^2\, dx+\Lambda\bigl||\Om_g|-d\bigr|\le F(\bar h, \bar\s, H_{\bar h, \bar\s})
$$
 and so by \eqref{lambda large}, $\bigl||\Om_g|-d\bigr|<\varepsilon$. In turn, by \eqref{graph above}, 
$$
d\le y_{max}\ell\le d+(1+c\ell)\varepsilon\,,
$$
which, again by 
\eqref{graph above}, yields
\begin{equation}\label{g max}
-c\varepsilon\le g(x)-\frac{d}{\ell}\le (1+c\ell)\varepsilon/\ell
\end{equation}
for all $x\in(0,\ell)$. It follows that $\|g-d/\ell\|_2\leq c\varepsilon$ for a possibly larger constant $c$ still independent of $g$. Hence,  using the minimality of $(g, \tau, H_{g, \tau})$ and \eqref{eq:beta}, we obtain  
\begin{align*}
\|\bar h- d/\ell\|_2&\leq \|\bar h- g\|_2+\|g-d/\ell\|_2\leq \sqrt{\frac1\beta F(\bar h, \bar \s, H_{\bar h,\bar  \s})}+\|g-d/\ell\|_2\nonumber\\
&<\frac12 \|\bar h- d/\ell\|_2+c\e\,,
\end{align*}
which is a contradiction if we choose $\e$ small enough.
 \end{proof}
Next we show that volume constrained minimizing configurations are also a unilateral minimizers of a simpler penalized  problem.
\begin{theorem}\label{sollevamento2} Let $d>0$ and let $(\bar{h},\bar{\sigma}, {H_{\bar h, \bar\s}})$ be a minimizing configuration
for problem \eqref{minG}. Fix $\Lambda> e_0^2W_0$. Then   $(\bar{h},\bar{\sigma}, {H_{\bar h, \bar\s}})$ is  a minimizer of
\beq\label{minGsemipen}
\min\biggl\{F(h,\sigma,H)+\Lambda\bigl(d-|\Om_h|\bigr):\,(h,\sigma,H)\in X(e_0;\mathbf{B}),\, |\Om_h|\leq  d\biggr\}\,.
\eeq
\end{theorem}
\begin{proof}
The proof is similar to the one of Step 1 of the 
proof of Theorem~\ref{sollevamento}, with $\beta=0$.
\end{proof}

The next lemma is proved in \cite[Lemma~6.5]{FuMo} and will be used to prove the interior ball condition stated in Lemma~\ref{pallaint} below. 
\begin{lemma}\label{crucialestimate}
Let $k\in AP(0,\ell)$ be nonnegative, let $B_\ro(z_0)$ be a ball such that $B_\ro(z_0)\subset 
\{(x,y): x\in (0,\ell)\text{ and } y<k(x)\}$, and   let $z_1=(x_1,y_1)$ and $z_2=(x_2,y_2)$ be points in
 $\pa B_\ro(z_0)\cap (\Gamma_k\cup\Sigma_k)$. Let $\gamma$ be the shortest arc on $\pa B_\ro(z_0)$ connecting
 $z_1$ and $z_2$ (any of the two possible arcs if $z_1$ and $z_2$ are antipodal) and 
 let $\gamma'$ be the arc on $\Gamma_k\cup\Sigma_k$ connecting $z_1$ and $z_2$.
 Then 
 $$
 \H^1(\gamma')-\H^1(\gamma)\geq \frac1\ro |D|\,,
 $$
 where $D$ is the region enclosed by $\gamma\cup\gamma'$.
\end{lemma}

\begin{lemma}\label{pallaint}
Let $\Lambda>0$ and let  $(g,\tau, H_{g, \tau})\in X(e_0;\B)$ be a minimizing configuration for  the problem 
\eqref{minGsemipen}.
 If  $\ro<\min\{1/\Lambda, r_0\}$, then  for all
 $z\in \Gamma_g\cup\Sigma_g$ there exists a ball $B_\ro(z_0)\subset\Om^\#_g\cup\big(
  \R\times (-\infty, 0]\big)$ such that $\pa B_\ro(z_0)\cap( \Gamma_g\cup\Sigma_g)=\{z\}$.
\end{lemma}
\begin{proof}
Fix $\ro<\min\{r_0, 1/\Lambda\}$. We argue by contradiction and assume that there exists $B_{\ro}(z_0)\subset\Om^\#_g\cup \bigl(\R\times (-\infty, 0]\bigr)$ touching $\widetilde \Gamma_g=\Gamma_g\cup\Sigma_g$ in at least two points $w_1=(s_1, t_1)$, $w_2=(s_2,t_2)\in S^+_\ro(z_0)$, where 
$S^+_\ro(z_0)$ denotes the upper half of $\pa B_{r_0}(z_0)$.  Consider the region $D$ bounded by the arc $\gamma$ on $S^+_\ro(z_0)$ connecting $w_1$ and $w_2$ and $\widetilde \Gamma_g$. Since $\ro<r_0$, necessarily 
$D\cap \cup_{i=1}^kB_{r_0}(z_i)=\emptyset$. 
 Hence  we may modify $g$ by replacing it with the function $\tilde g$ which coincides with $g$ in 
$[0,\ell)\setminus(s_1,s_2)$ and whose graph on $(s_1,s_2)$ is given by $\gamma$. Denote 
by $\gamma'$ the arc on $\widetilde \Gamma_g$ connecting $w_1$ and $w_2$. Then
we have
$$
F(\tilde g, v)+\Lambda\bigl(d-|\Om_{\tilde g}|\bigr)-F( g, v)-\Lambda\bigl(d-|\Om_{ g}|\bigr)
\leq \H^1(\gamma)-\H^1(\gamma')+\Lambda|D|<0\,,
$$
where the last inequality is a consequence of  Lemma~\ref{crucialestimate} and the 
fact that  $\ro<1/\Lambda$. This contradicts the minimality of $(g,\tau, H_{g,\tau})$. The conclusion of the lemma follows arguing as in \cite[Lemma 2]{CL} or \cite[Proposition 3.3, Step 2]{FFLM}.
\end{proof}

Theorem~\ref{sollevamento} will be used to study the regularity of those profiles for which the function $h^-$ defined in \eqref{piuomeno} is not flat. Note the assumption that $h^-$ is flat does not exclude a priori the presence of vertical cuts (see \eqref{Sigma_g}). This possibility is ruled out by the next result.

\begin{theorem}\label{th:regflat}
Let $(\bar h,\bar \sigma,H_{\bar h, \bar\sigma})$ be a minimizing configuration of problem \eqref{minG} such that $\bar h^-$ is constant. Then $\Sigma_{\bar h}=\emptyset$.
\end{theorem}
\begin{proof}
By Theorem~\ref{sollevamento2} and Lemma~\ref{pallaint} we deduce that  $\Om^\#_{\bar h}\cup\bigl(\R\times (-\infty, 0]\bigr)$ satisfies an interior ball condition with $\ro<\min\{r_0, 1/(e_0^2W_0)\}$. If $\Sigma_{\bar h}$ were nonempty, then each vertical cut would meet the (horizontal) graph of $\bar h^-$ perpendicularly, but this would prevent the existence of an interior sphere at the corner.  Hence, $\Sigma_{\bar h}=\emptyset$ and the proof is complete.
\end{proof}

We now recall some regularity estimates, based on the theory developed by Grisvard (\cite{Grisvard}),  proved in \cite{FFLM} for solutions of the Lam\'e system in planar  domains with a corner.

Let $\Omega$ be a bounded open set in $\mathbb{R}^{2}$ whose boundary can be
decomposed in three curves
\[
\partial\Omega=\Gamma_{1}\cup\Gamma_{2}\cup\Gamma_{3},
\]
where $\Gamma_{1}$ and $\Gamma_{2}$ are two segments meeting at the origin
with an (internal) angle $\omega\in\left(\pi,2\pi\right)  $ and $\Gamma_{3}$
is a smooth curve joining the two remaining endpoints of $\Gamma_{1}$ and
$\Gamma_{2}$ in a smooth way and not passing through the origin.
We shall refer to such an open set as a {\em regular domain with corner angle $\omega$}.

The next result is a particular case of \cite[Th\'eor\`eme I]{Grisvard}. 

\begin{theorem}
\label{theorem grisvard}Let $\Omega\subset\mathbb{R}^{2}$ be  a regular domain with corner angle $\omega\in (\pi, 2\pi)$ and  let
$w\in H^1(\Omega;\mathbb{R}^{2})  $ be a weak solution of the Neumann problem
\beq\label{grisvard1}
\begin{cases}
\Div\, \C E(w)= f & \text{\quad
in }\Omega,\\
\C E(w)[\nu]= g & \quad\text{on }\partial\Omega,
\end{cases}
\eeq
where $f\in L^{p}\left(  \Omega;\mathbb{R}^{2}\right)$ and $g\in W^{1-1/p, p}(\pa \Om\setminus\{0\}; \R^2)$, $p\in (1,2)$. Then,  there exist numbers
$c_\alpha$, $c'_\alpha$ such that 
$w$ may be decomposed as
$$
w=w_{\operatorname*{reg}}+\sum_{\alpha}c_{\alpha}%
\mathcal{S}_{\alpha}+\sum_{\alpha}c'_\alpha\frac{\pa }{\pa\alpha} \mathcal{S}_{\alpha}, \label{decomposition}%
$$
where $w_{reg}\in W^{2,p}(\Om; \R^2)$ and in the first sum $\alpha$ ranges among all complex numbers with $\operatorname*{Re} \alpha\in\left(  0,\frac{2(p-1)}{p}\right)$ which are solutions of the equation
\begin{equation}
\sin^{2}\alpha\omega=\alpha^{2}\sin^{2}\omega, \label{trascendente}%
\end{equation}
and in the second sum $\alpha$ ranges only among solutions with multiplicity two  of \eqref{trascendente} in the same strip.
Moreover, the functions $\mathcal{S}_{\alpha}$ are independent of ${f}$ and in
polar coordinates
\[
\mathcal{S}_{\alpha}\left(  r,\theta\right)  =r^{\alpha}{g}_{\alpha
}\left( \theta\right)\,,
\]
for some smooth function $g_\alpha$. The above decomposition holds provided that \eqref{trascendente} has no solutions with real part equal to
$\frac{2(p-1)}{p}$. 

\end{theorem}

Though this  result gives no information about the roots of  equation \eqref{trascendente}, it is clear that  the solutions  contained in the strip $0<{\rm Re}\,\alpha<1$ are bounded. Hence, by analyticity, they are finitely many. A more precise information is provided by the following result, proved in \cite[Theorem 2.2]{Nicaise}.
\begin{theorem}\label{nicaise}
If $\om\in(0,2\pi)$, then equation  \eqref{trascendente} has no roots in the strip $0<{\rm Re}\,\alpha\leq\displaystyle\frac12$. 
\end{theorem}
We will use the two previous results to get an a priori estimate for the solutions  to \eqref{grisvard1}. We recall that an {\it infinitesimal rigid motion} is an affine displacement of the form $a+Ax$, where $A$ is a skew symmetric $2\times2$ matrix and $a$ is a constant vector.
\begin{proposition}\label{g+n}
Let $\Om$ be as in Theorem~\ref{theorem grisvard}. There exist $p\in(4/3,2)$ and $C>0$ such that if $f\in L^p(\Om;\R^2)$, $g\in W^{1-1/p, p}(\pa \Om\setminus\{0\}; \R^2)$ and $w\in 
W^{1,2}
(\Om;\R^2)$ is a weak solution to  problem \eqref{grisvard1}, then
\begin{equation}\label{g+n1}
\|w\|_{W^{2,p}(\Om;\R^2)}\leq C\bigl(\|w\|_{L^{p}(\Om;\R^2)}+\|f\|_{L^{p}(\Om;\R^2)}+\|g\|_{W^{1-1/p, p}(\pa \Om\setminus\{0\}; \R^2)}\bigr)\,.
\end{equation}
\end{proposition}
\begin{proof} As observed above, the strip $0<{\rm Re}\,\alpha<1$ contains only finitely many solutions to equation \eqref{trascendente}. Hence, by Theorem~\ref{nicaise} there exists $\e>0$ such that all solutions are contained in the strip $\frac12+\e<{\rm Re}\,\alpha<1$. Therefore, if we choose $p>4/3$ such that $2-\frac2p<\frac12+\e$, from Theorem~\ref{theorem grisvard} we get that any weak solution to \eqref{grisvard1}, with $f\in L^p(\Om;\R^2)$ and $g\in W^{1-1/p, p}(\pa \Om\setminus\{0\}; \R^2)$ is in $W^{2,p}(\Om;\R^2)$.

To prove \eqref{g+n1}, set $V:=W^{2,p}(\Om;\R^2)/\thicksim$, where for every $u,v\in W^{2,p}(\Om; \R^2)$, we have set $u\thicksim v$ if and only if $u-v$ is an infinitesimal rigid motion. We define a norm in ${V}$ setting
$$
\|[u]\|_{V}:=\|E(u)\|_{L^p(\Om;\R^2)}+\|\nabla^2u\|_{L^p(\Om)}
$$
 for every equivalence class $[u]$, with $u\in W^{2,p}(\Om;\R^2)$. 
Note that this definition is well posed, since if $u\thicksim v$, then $E(u)=E(v)$ and $\nabla^2u=\nabla^2v$.
 Note also that in view of Korn's inequality, $V$ is a Banach space.  

Consider now the operator $L:{ V}\to L^p(\Om;\R^2)\times W^{1-1/p, p}(\pa \Om\setminus\{0\}; \R^2)$ defined for any $[u]\in{V}$ as
$$
L[u]:=(\Div\, \C E(u), \C E(u)[\nu])\,.
$$
By the first part of the proof we have that $L$ is a linear, continuous, and invertible operator between two Banach spaces. Therefore, the conclusion follows from  the open mapping theorem.
\end{proof}

\begin{proposition}\label{2decays}
Let $\Om$ be a regular domain with corner $\omega\in (\pi, 2\pi)$  and let $u\in H^1(\Om;\R^2)$ be   a weak solution to the Neumann problem 
$$
\begin{cases}
\Div\, \C E(w)= f & \text{\quad in }\Omega,\\
\C E(w)[\nu]= g & \quad\text{on }\Gamma_1\cup\Gamma_2,
\end{cases}
$$
with $f\in L^p(\Om; \R^2)$ and $g\in W^{1-1/p, p}((\Gamma_1\cup\Gamma_2)\setminus\{0\}; \R^2)$.
Then, there exist   $\bar r>0$, with  ${\overline B}_{\bar r}(0)\cap\Gamma_3=\emptyset$,   $C>0$, and $\alpha>1/2$, depending only on $\lambda,\mu$, $\omega$, $\|f\|_{L^p(\Om; \R^2)}$ and $\|g\|_{W^{1-1/p, p}((\Gamma_1\cup\Gamma_2)\setminus\{0\}; \R^2)}$, such that for all $r\in(0,\bar r)$,
\begin{equation}\label{2decays1}
\int_{B_r(0)\cap\Om}|\nabla w|^2\,dz\leq Cr^{2\alpha}\int_\Om\bigl(1+|w|^2+|\nabla w|^2\bigr)\,dz\,.
\end{equation}
\end{proposition}
\begin{proof} Set $B_{\hat r}:=B_{\hat r}(0)$ and fix $\hat r>0$ such that $B_{\hat r}\cap \Gamma_3=\emptyset$ and $\pa B_{\hat r}\cap \Gamma_1\cup\Gamma_2\neq \emptyset$, and $0<\bar r<\hat r$. 
Let $\vphi\in C^{\infty}_c(B_{\hat r})$ be such that $\vphi\equiv 1$ on $B_{\bar r}$. From the equation satisfied by $w\vphi$ and from 
\eqref{g+n1} we get 
\begin{align*}
\|w\vphi\|_{W^{2,p}(\Om;\R^2)}&\leq C\bigl(\|w\|_{W^{1,p}(\Om;\R^2)}+\|f\|_{L^{p}(\Om;\R^2)}\\
&\quad\qquad+\|g\|_{W^{1-1/p, p}((\Gamma_1\cup\Gamma_2)\setminus\{0\}; \R^2)}+ 
\|w\|_{W^{1-1/p, p}(\pa \Om\setminus\{0\}; \R^2)}\bigr)\\
&\leq C\bigl(\|w\|_{W^{1,p}(\Om;\R^2)}+\|f\|_{L^{p}(\Om;\R^2)}+\|g\|_{W^{1-1/p, p}((\Gamma_1\cup\Gamma_2)\setminus\{0\}; \R^2)}\bigr)
\end{align*}
for some $\frac43<p<2$ and some $C>0$ depending only on $\lambda$, $\mu$ and $\omega$.
 Thus, if $0<r<\bar r$, using the Sobolev imbedding theorem we have
\begin{align*}
\int_{B_r\cap\Om}|\nabla w|^2\,dz 
& \leq c\biggl(\int_{B_r\cap\Om}|\nabla (w\vphi)|^{\frac{2p}{2-p}}\,dz\biggr)^{\frac{2-p}{p}}r^{\frac{4(p-1)}{p}}  \leq cr^{\frac{4(p-1)}{p}}\|w\vphi\|^2_{W^{2,p}(\Om;\R^2)}\\
&\leq cr^{\frac{4(p-1)}{p}}\bigl(1+\|w\|_{W^{1,p}(\Om;\R^2)}\bigr)^2\leq c r^{2\alpha} \int_\Om\bigl(1+|w|^2+|\nabla w|^2\bigr)\,dz\,,
\end{align*}
where $\alpha:=2(p-1)/p$ is strictly greater than $1/2$ since $p>4/3$.
\end{proof}
For $g\in AP(0,\ell)$  we denote the set of {\it cusp points} by 
$$
 \Sigma_{g,c}:=\{(x,g(x)):\,  x\in [0,\ell)\,,  g^-(x)=g(x)\,,\text{ and } g^\prime_+(x)=-g^\prime_-(x)=+\infty\}\,,
$$
where $g^-$ is defined in \eqref{piuomeno}, while $g^\prime_+$ and $g^\prime_-$ denote the right and left derivatives, respectively. 

As usual, the set $\Sigma^\#_{g,c}$ is obtained by replacing $[0,\ell)$ by $\R$ in the previous formula and coincides with the $\ell$-periodic extension of $\Sigma_{g,c}$. 

\begin{theorem}[Regularity]\label{th:regolarita}
Let $(\bar h,\bar \sigma, H_{\bar h,\bar\s})\in  X(e_0;\B)$ be a  minimizer of \eqref{minG}, with $\bar \s=\sum_{i=1}^k\bu_i\delta^\#_{z_i}$.  Then:
\begin{itemize}
\item[(i)]  $\bar h$ has  at most finitely many cusp points and vertical cracks  $[0,\ell)$, i.e.,
$$
\operatorname*{card\,}\bigl(\{x\in [0,\ell):\, (x,y)\in\Sigma_{\bar h}\cup \Sigma_{\bar h,c}\,\text{ for some } y\geq 0\}\bigr)<+\infty\,;
$$
\item[(ii)] the curve $\Gamma_{\bar h}^\#$ is of class $C^{1}$ away from 
$\Sigma^\#_{\bar h}\cup \Sigma^\#_{\bar h,c}$ and 
$$
\lim_{x\to x_0^\pm}\bar h'(x)=\pm \infty\qquad\text{for every $x_0$ s.t. $(x_0, \bar h(x_0))\in \Sigma^\#_{\bar h}\cup \Sigma^\#_{\bar h,c}$;}
$$
\item[(iii)] $\Gamma_{\bar h}^\#\cap \{y>0\}$ is of class $C^{1,\alpha}$  away from 
$\Sigma^\#_{\bar h}\cup \Sigma^\#_{\bar h,c}$ for all $\alpha\in (0,1/2)$;
\item[(iv)]  setting  
$$
A:=\{(x,y)\in \R^2:\, \bar h(x)>0,\, \bar h\text{ continuous at }x\}\,,
$$ 
$\Gamma^\#_{\bar h}$ is analytic in $A\setminus \cup_{i=1}^k\cup_{m\in \Z}\overline B_{r_0}(z_i+m\ell\mathbf{e}_1)$. 
\end{itemize}
\end{theorem}

The proof of the regularity theorem is based upon the strategy introduced in \cite{FFLM} (see also \cite{FuMo}).  We only outline the main steps, by highlighting the changes needed in the present situation and referring the reader to the aforementioned papers for the details.

\begin{proof}[Proof of Theorem~\ref{th:regolarita}]
We start by observing that we may assume that $\bar h^-$ is not constant, since otherwise the conclusion  follows from Theorem~\ref{th:regflat}. Note also the if $d<2r_0\ell$, then necessarily $\B=\emptyset$, and thus the result follows from \cite[Theorem~2.5]{DF} (see also \cite[Theorem~2.7]{FuMo}). 
Therefore, from now on we shall assume that $d\geq 2r_0\ell$ and $\bar h^-$ is not constant. 

\noindent{\bf Step 1.} (Lipschitz partial regularity) From Theorem~\ref{sollevamento2} and Lemma~\ref{pallaint} we have that 
$\widetilde\Gamma_{\bar h}$ satisfies an interior ball condition with radius $\ro<\min\{1/(e_0^2W_0), r_0\}$.  By applying 
\cite[Lemma~3]{CL} we get that $\widetilde\Gamma_{\bar h}$ has the following properties: For any $z_0 \in\widetilde{\Gamma}_{\bar h}$ there exist an orthonormal basis $\mathbf{i}$,
$\mathbf{j}\in\mathbb{R}^{2}$ and a rectangle
\[
Q:=\{{z}_{0}+s\mathbf{i}+t\mathbf{j}:\,-a^{\prime}<s<a^{\prime
},\,-b^{\prime}<t<b^{\prime}\},
\]
with $a^{\prime}$, $b^{\prime}>0$, such that $\Omega_{\bar h}\cap Q$ has one of the
following two representations:

\noindent (j) There exists a Lipschitz function $f:(-a^{\prime},a^{\prime
})\rightarrow(-b^{\prime},b^{\prime})$ such that $f\left(  0\right)  =0$ and
\[
\Omega_{\bar h}\cap Q=\{{z}_{0}+s\mathbf{i}+t\mathbf{j}:\,-a^{\prime
}<s<a^{\prime},\,-b^{\prime}<t<f(s)\}\cap\left(  \left(  0,\ell\right)
\times\mathbb{R}\right)  .
\]
Moreover, the function $f$ admits at every point left and right derivatives, which  are left and right continuous, respectively.

\noindent (jj) There exist two Lipschitz functions $f_{1}$, $f_{2}:\left[
0,a^{\prime}\right)  \rightarrow(-b^{\prime},b^{\prime})$ such that
$f_{i}\left(  0\right)  =\left(  f_{i}\right)  _{+}^{\prime}\left(  0\right)
=0$ for $i=1,$ $2,$ $f_{1}\leq f_{2}$, and%
\[
\Omega_{\bar h}\cap Q=\{{z}_{0}+s\mathbf{i}+t\mathbf{j}:\,0<s<a^{\prime
},\,-b^{\prime}<t<f_{1}(s) \text{\, or \,} f_{2}(s)<t<b^{\prime}\}\,.
\]
Moreover, the functions $f_{1},$ $f_{2}$ admit at every point left and right derivatives, which  are left and right continuous, respectively.
Note that (j) and (jj) imply statement (i) of the theorem and the fact that
$$
\lim_{x\to x_0^\pm}\bar h_\pm'(x)=\pm \infty\qquad\text{ for every $x_0$ s.t. $(x_0, \bar h(x_0))\in \Sigma_{\bar h}\cup \Sigma_{\bar h,c}$.}
$$

\noindent{\bf Step 2.} ($C^1$-regularity) From property (j) of Step 1 we have that the curve $\Gamma_{\bar h}$ is locally Lipschitz in $[0, \ell)\times \R$ away from finitely many  singularities  of cusp or cut type. Moreover, outside the singular set,  $\Gamma_{\bar h}$  admits left and right tangent, which are left and right continuous respectively. Therefore, to prove statement (ii) it is enough to show that left and right tangents coincide at  every point $z_0\not\in \Sigma_{\bar h}\cup \Sigma_{\bar h,c}$.

 Assume by contradiction that this does not happen for some $z_0=(x_0, y_0)\not\in \Sigma_{\bar h}\cup \Sigma_{\bar h,c}$.  
 If $y_0=0$, then by interior ball condition we can say that there are no dislocation balls in a neighborhood $B_{r}(z_0)$ of $z_0$ and thus $H_{\bar h, \bar \s}$ in such a neighborhood is a gradient $Dv$, with $v$ satisfying
 $$
\begin{cases}
\Div\, \C E(v)=0 & \text{in $\Om_{\bar h}\cap B_{r}(z_0)$,}\\
\C E(v)[\nu]=0 & \text{on $\Gamma_{\bar h}\cap B_{r}(z_0)$,}\\
v(x,0)=e_0(x,0) & \text{on $\{y=0\}\cap B_{r}(z_0)$.}
\end{cases}
$$
We may therefore apply the argument used in \cite[Theorem~4.9 and Proposition~5.1]{DF} to obtain a contradiction.

 Assume now that $y_0>0$. In this case we decompose $H_{\bar h, \bar \s}= Dv+K$, where 
 $$
K:=
\left(
\begin{array}{cc}
k_1 & 0\\
k_2 & 0
\end{array}
\right)\,,
$$
with 
$$
k_l(x,y):=-\sum_{i=1}^k(\bu_i\cdot \mathbf{e}_l)\int_0^y\ro_{r_0}(x-x_i, t-y_i)\, dt \text{ for $l=1,2$,}
$$
and $v$ satisfies
$$
\begin{cases}
\Div\, \C E(v)=-\Div\, \C K_{sym} & \text{in $\Om_{\bar h}$,}\\
\C E(v)[\nu]=- \C K_{sym}[\nu] & \text{on $\Gamma_{\bar h}$.}
\end{cases}
$$
Using \eqref{2decays1} in place of \cite[Equation (3.52)]{FFLM} and arguing as in \cite[Theorem~3.13]{FFLM},  we can prove that there exist $C>0$, a radius $\bar r>0$, and $\alpha\in (\frac12,1)$ such that 
$$
\int_{B_r(z_0)\cap\Om^\#_{\bar h}}|\nabla v|^2\, dz\leq C r^{2\alpha}\qquad\text{for all $r\leq\bar r$.}
$$
 In turn, since $K$ is smooth this implies that for a possibly larger constant
$$
\int_{B_r(z_0)\cap\Om^\#_{\bar h}}|H_{\bar h, \bar\s}|^2\, dz\leq C r^{2\alpha}\qquad\text{for all $r\leq\bar r$.}
$$
Moreover, by Theorem~\ref{sollevamento} there exist $\Lambda$, $\beta>0$
such that $(\bar h, \bar \s, H_{\bar h,\bar\s})$ is a minimizer of 
\beq\label{minGpen2}
\min\biggl\{F(h,\sigma,H)+\beta\int_0^\ell|h-\bar h|^2\, dx+\Lambda\bigl||\Om_h|-d\bigr|:\,(h,\sigma,H)\in X(e_0;\mathbf{B})\biggr\}\,.
\eeq
 To fix the ideas let us assume that $z_0=(x_0, \bar h(x_0))$ does not belong to a vertical segment of $\Gamma_{\bar h}$; i.e., $\bar h$ is continuous at $x_0$. The other case can be dealt with similarly. 
 
Observe that by a standard extension argument we may
define $v$ in a fixed neighborhood of ${z}_{0}$ in such a way
that, denoting by $\tilde v$ the resulting function, for all $0<r\leq \bar r$ we have
$$
\int_{B_r({ {z}_{0}})  }|\nabla \tilde v|^{2}\,dz%
\leq c(L)\int_{B_r({ {z}_{0}})  \cap\Omega_{\bar h}}|\nabla
v|^{2}\,dz\,, 
$$
where the constant $c(L)$ 
depends only on the Lipschitz constant $L$ of the function $\bar h$. Finally set $H:=D\tilde v+ K$ and observe that
\beq\label{reg3}
\int_{B_r(z_0)}|H|^2\, dz\leq C r^{2\alpha}\qquad\text{for all $r\leq\bar r$.}
\eeq
For $r>0$\ (sufficiently small) we denote
\[
x_{r}^{\prime}:=\max\{x\in (0,\ell):\,x\leq x_{0}\text{ and there exists $y$ such that
}(x,y)\in\Gamma_{\bar h}\cap\partial B_{r}({z}_{0})  \}\,,
\]%
\[
x_{r}^{\prime\prime}:=\min\{x\in(0,\ell):\,x\geq x_{0},\text{ and there exists
$y$ such that }(x,y)\in\Gamma_{\bar h}\cap\partial B_{r}({z}_{0})  \}\,,
\]
and we let $(x_{r}^{\prime},\bar h(x_{r}^{\prime}))$ and $(x_{r}^{\prime\prime
},\bar h(x_{r}^{\prime\prime}))$ be the corresponding points on $\Gamma_{\bar h}\cap\partial
B_{r}(  {z}_{0})$. Construct $h_{r}$ as the greatest lower
semicontinuous function coinciding with $\bar h$ outside $[x_{r}^{\prime}%
,x_{r}^{\prime\prime}]$ and with the affine function
\[
x\mapsto \bar h(x_{r}^{\prime})+\frac{\bar h(x_{r}^{\prime\prime})-\bar h(x_{r}^{\prime}%
)}{x_{r}^{\prime\prime}-x_{r}^{\prime}}(x-x_{r}^{\prime})
\]
in $(x_{r}^{\prime},x_{r}^{\prime\prime})$. For $r>0$
sufficiently small $(h_r, \bar\s, H)$ is admissible for the
penalized minimization problem \eqref{minGpen2}  . Hence,
\[
F(\bar h, \bar\s, H_{\bar h, \bar \s})\leq F(h_r,\bar \sigma,H)+\beta\int_0^b|h_r-\bar h|^2\, dx+\Lambda\bigl||\Om_{h_r}|-d\bigr|\,.
\]
Since $h_r=\bar h$ outside $[x_{r}^{\prime}%
,x_{r}^{\prime\prime}]$ and $H=H_{\bar h,\bar\sigma}$ outside $B_{r}(z_0)$, using  \eqref{reg3}, we get
\begin{equation}%
\begin{array}
[c]{rl}%
\displaystyle  \int_{x_{r}^{\prime}}^{x_{r}^{\prime\prime}}\sqrt{1+(\bar h^{\prime
})^{2}}\,dx & \displaystyle  \leq\int_{x_{r}^{\prime}}^{x_{r}^{\prime\prime}%
}{\textstyle  \sqrt{1+(h_{r}^{\prime})^{2}}}\,dx+C\int_{B_{r}(z_0)}|H|^{2}\,dz+Cr^{2}\\
& \displaystyle  \leq\int_{x_{r}^{\prime}}^{x_{r}^{\prime\prime}}{\textstyle
\sqrt{1+(h_{r}^{\prime})^{2}}}\,dx+Cr^{2\alpha}%
\end{array}
\label{estimate h}%
\end{equation}
for $r$ small enough. On the other hand, since 
the right and the left derivatives $\bar h_{+}^{\prime}$ and $\bar h_{-}^{\prime}$ exist and are
continuous in a neighborhood of $x_{0}$, it can be checked  that (see \cite[Proof of Theorem~3.14]{FFLM}) 
$$
\displaystyle  \int_{x_{r}^{\prime}}^{x_{r}^{\prime\prime}}\sqrt{1+(\bar h^{\prime
})^{2}}\,dx- \int_{x_{r}^{\prime}}^{x_{r}^{\prime\prime}%
}{\textstyle  \sqrt{1+(h_{r}^{\prime})^{2}}}\,dx\geq C_0 r
$$
for $r$ sufficiently small, where $C_0>0$ depends only on the angle at the corner point $z_0$. Since $2\alpha>1$ this contradict \eqref{estimate h}.

\noindent{\bf Step 3.} ($C^{1,\alpha}$-regularity) 
 Fix an open subarc $\Gamma\subset\Gamma_{\bar h}\setminus (\Sigma_{\bar h}\cup \Sigma_{\bar h,c})$ not intersecting $\{y=0\}$. 
 As in Step 2, we consider only the case in which $\Gamma$  does not contain vertical parts, the other case being analogous. Let $I$ be the projection of $\Gamma$ onto the $x$-axis. By taking $\Gamma$ smaller, if needed, we may assume that $I\times(0,\infty)$ intersects at most one ball $B_{r_0}(z_j)$, $j=1,\dots, k$ and, by Step 2, that  $\bar h\in C^1(\bar I)$.  
Fix $J\subset\subset I$ and consider the decomposition of $H_{\bar h, \bar \s}$ introduced in Step 2. For any $\alpha \in (0,1)$ there exist 
$C$, $\bar r>0$ such that if $z_0=(x_0, \bar h(x_0))$, $x_0\in J$, then 
$$
\int_{B_r(z_0)\cap\Om^\#_{\bar h}}|\nabla v|^2\, dz\leq C r^{2\alpha}\qquad\text{for all $r\leq\bar r$}\,.
$$
Such a decay estimate can be established exactly as in \cite[Theorem~3.16]{FFLM}. Note that both $C$ and $\bar r$ are uniform with respect 
to $x_0\in J$.  Arguing as in the previous step, we  may extend $H_{\bar h, \bar\s}$ to $B_{\bar r}(z_0)$ in such a way that the resulting strain field $H$ satisfies
\beq\label{reg4}
\int_{B_r(z_0)}|H|^2\, dz\leq C r^{2\alpha}\qquad\text{for all $r\leq\bar r$,}
\eeq
for a possibly larger constant  $C$ still independent of $z_0$.    Fix $r<\bar r$ and consider the affine function $s$  connecting $z_0$ and $(x_0+r, \bar h(x_0+r))$.
If the graph of $s$ over the interval   $(x_0, x_0+r)$ does not intersect any of the balls $B_{r_0}(z_j)$, $j=1,\dots, k$, we can proceed as in \cite[Step 5 of the proof of Theorem~6.9]{FuMo}. Thus assume that the graph of $s$  over the interval   $(x_0, x_0+r)$ intersects a ball $B_{r_0}(z_j)$. Note that by construction of $I$ there can only be one such ball.
Define $h_r$ as 
 $$
 h_r(x):=
 \begin{cases}
 \bar h(x) & \text{$x\in [0, \ell)\setminus(x_0,x_0+r)$,}\\
 \max\{f_j(x), s(x)\} & \text{$x\in [x_0, x_0+r]$,}
 \end{cases}
 $$ 
 where $f_j(x):=y_j+\sqrt{r_0^2-(x-x_j)^2}$. Note that $(h_r, \bar \s, H)$ is admissible for problem \eqref{minGpen2}. Then using the minimality
 of $(\bar h, \bar \s, H_{\bar h, \bar s})$, the decay estimate \eqref{reg4}, and 
 arguing as in Step \textcolor{red}{2} we obtain 
 $$
 \int_{x_0}^{x_0+r}\sqrt{1+ (\bar h^{\prime})^2}\, dx\leq C r^{2\alpha}+ \int_{x_0}^{x_0+r}\sqrt{1+ (h_r^{\prime })^2}\, dx\,,
 $$
 for some constant $C$ independent of $x_0\in J$. This inequality
 can be equivalently written as
 \begin{eqnarray}\label{ineq-tost}
&& \int_{x_0}^{x_0+r}\sqrt{1+ (\bar h^{\prime})^2}\, dx- \sqrt{(\bar h(x_0+r)-\bar h(x_0))^2 +r^2}\nonumber \\
&&\qquad\qquad\leq 
 C r^{2\alpha}+ \int_{x_0}^{x_0+r}\Bigl(\sqrt{1+ (h_r^{\prime })^2}-\sqrt{1+(s^{\prime })^2}\Bigr)\, dx\nonumber\\
 &&\qquad\qquad=  C r^{2\alpha}+ \int_{(x_0, x_0+r)\cap\{f_j>s\}}
\!\Bigl(\sqrt{1+(f_j^{\prime })^2}-\sqrt{1+(s^{\prime })^2}\Bigr)\, dx\nonumber\\
&&\qquad\qquad =C r^{2\alpha}+ \int_{(x_0, x_0+r)\cap\{f_j>s\}}
\!\Bigl(\sqrt{1+(f_j^{\prime })^2}-\sqrt{1+(f_j^{\prime }(\bar x))^2}\Bigr)\, dx \leq C' r^{2\alpha}\,.
 \end{eqnarray}
Note that in the second equality  we used the fact that since $\bar h\ge f_j$ and the graph of $s$ joins two points of the graph of $\bar{h}$, it must intersect the graph of $f_j$ twice. Hence, by the mean value theorem we may find $\bar x\in 
 (x_0, x_0+r)\cap \{f_j>s\}$  such that $f_j'(\bar x)=s'(\bar{x})$. In the last inequality  we used the fact that $f_j'$ is Lipschitz.
 On the other hand, using the inequality
 $$
 \sqrt{1+b^2}-\sqrt{1+a^2}\geq \frac{a(b-a)}{\sqrt{1+a^2}}+\frac{(b-a)^2}{2(1+\max\{a^2,b^2\})^{3/2}}
 $$
 with $a:=\medintinrigo_{x_0}^{x_0+r}\bar h'\, dx$ and
 $b:=\bar h'(x)$, and integrating the result in $(x_0, x_0+r)$, we get
\begin{eqnarray*}
 &&\frac{1}{2(1+M^2)^{3/2}}
 \medint_{x_0}^{x_0+r}\Bigl(\bar h'(x)- \medint_{x_0}^{x_0+r}\bar h'\, ds\Bigr)^2\, dx\\
 && \qquad\qquad\qquad\leq 
 \frac1r \int_{x_0}^{x_0+r}\sqrt{1+ \bar h^{\prime 2}}\, dx- \frac1r\sqrt{(\bar h(x_0+r)-\bar h(x_0))^2 +r^2}\leq 
 C' r^{2\alpha-1}\,,
 \end{eqnarray*}
 where $M$ is the Lipschitz constant of $\bar h$ in $I$ and we used \eqref{ineq-tost}.
 In particular, 
 $$
 \medint_{x_0}^{x_0+r}\Bigl|\bar h'(x)- \medint_{x_0}^{x_0+r}\bar h'\, ds\Bigr|\, dx
 \leq C'' r^{\alpha-\frac12}\,.
 $$
 A similar inequality holds also in the interval $(x_0-r, x_0)$. Hence, by the arbitrariness of $x_0\in J$ and  
 \cite[Theorem 7.51]{AFP} we conclude that  $\bar h\in C^{1,\alpha-\frac12}(J)$ for all $\alpha\in (1/2,1)$, as claimed. 
This concludes the proof of statement (iii) of the theorem. 

\noindent{\bf Step 4.} To prove the analytic regularity, observe that in  $A\setminus \cup_{i=1}^k\cup_{m\in \Z}\overline B_{r_0}(z_i)$ we can perform variations of the profile $\bar h$ to prove that  \eqref{eq:EL} holds in the weak sense. Thus, in particular,
the curvature $\kappa$ is of class $C^{0,\alpha}$ in $A\setminus \cup_{i=1}^k\cup_{m\in \Z}\overline B_{r_0}(z_i)$ for all $\alpha\in (0,\frac12)$.   A standard bootstrap argument implies the $C^{\infty}$-regularity. Analyticity then follows from Theorem~4.9 and the remarks at the end of Section~4.2 in \cite{KLM}.
\end{proof}

\subsection{Dislocations accumulate at the bottom}
In this subsection we consider  nearly flat  profiles $h$. We will show that if $e_0$ is sufficiently large and  $(h, \sigma, H)$ is any  admissible configuration  in $X (e_0; \mathbf{B})$, then, by moving the dislocations centers of $\sigma$ in the direction $-\mathbf{e_2}$, the elastic energy decreases.
This is made precise by the following proposition.
\begin{proposition}\label{prop:nicebottom}
Given $d>2r_0\ell$ and $\alpha\in (0,1)$, there exist $\overline e>0$ and $\de>0$ such that if $e_0 (\bu_i\cdot\mathbf{e}_1)>0$ for all $\bu_i\in \mathbf{B}$, $i=1,\dots, k$ and $|e_0|>\overline e$,  then for every  $(h, \sigma, H_{h, \s})\in X(e_0; \mathbf{B})$, with $\|h-d/\ell\|_{C^{1,\alpha}_\#(0,\ell)}\leq\de$ and $\sigma=\sum_{i=1}^k\bu_i\delta^\#_{z_i}$, with $z_j\cdot \mathbf{e}_2>0$ for some $j\in\{1,\dots, k\}$, we have
$$
\int_{\Om_h}W((H_{h, \s_s})_{sym})\, dz<\int_{\Om_h}W((H_{h, \s})_{sym})\, dz
$$
for all $s>0$ sufficiently small, where $\sigma_s:=\sum_{i=1, i\neq j}^k\bu_i\delta^\#_{z_i}+\bu_j
\delta^\#_{z_j-s\mathbf{e_2}}$. In particular, if $(h, \sigma, H_{h,\sigma})$ is a minimizer of \eqref{minG}, then all dislocations lie at the bottom of $\Om_h$, that is all the centers $z_i$ are of the form $z_i=(x_i, r_0)$.
\end{proposition}
\begin{proof}
Assume, without loss of generality, that $e_0>0$. It is enough to show that for $e_0$ sufficiently large and $\de$ small
\beq\label{derG}
\frac{d}{ds}F(h,  \sigma_s, H_{h, \s_s})_{\bigl|_{s=0}}>0\,,
\eeq
where $(h, \sigma, H_{h, \s})\in X(e_0; \mathbf{B})$ is as in the statement.
For simplicity set 
$$
H_s:=H_{h, \s_s} \qquad\text{and}\qquad H:=H_0=H_{h, \s}\,.
$$
and recall that, by \eqref{ELHhs}, $H_s$ is the unique periodic solution to the following system
$$
\begin{cases}
\displaystyle\curl H_s=\sum_{i=1, i\neq j}^k\bu_i\ro^\#_{r_0}(\cdot-z_i)+\bu_j\ro^\#_{r_0}(\cdot-z_j-s\mathbf{e_2}) & \text{in $\Om_h$,}\\
\Div\, \C(H_s)_{sym}=0& \text{in $\Om_h$,}\vspace{5pt}\\
\C(H_s)_{sym}[\nu]=0 &\text{on $\Gamma_h$,}\vspace{5pt}\\
H_s[\,\mathbf{e}_1\,]=e_0\mathbf{e}_1 & \text{on $\{y=0\}$.}
\end{cases}
$$
Then the derivative in \eqref{derG} reduces to
\beq\label{schifo-1}
\frac{d}{ds}\biggl(\frac12\int_{\Om_h}\C(H_s)_{sym}:(H_s)_{sym}\, dz\biggr)_{\bigl|_{s=0}}=\int_{\Om_h}\C H_{sym}: \dot H_{sym}\, dz\,,
\eeq
where $\dot{H}=\frac{d}{ds}H_{s_{|_{s=0}}}$ is determined as the unique periodic solution to 
\begin{equation}\label{200}
\begin{cases}
\curl \dot{H}=-\bu_jD_y\ro^\#_{r_0}(\cdot-z_j) & \text{in $\Om_h$,}\\
\Div\, \C\dot{H}_{sym}=0& \text{in $\Om_h$,}\\
\C\dot{H}_{sym}[\nu]=0 &\text{on $\Gamma_h$,}\\
\dot{H}[\,\mathbf{e}_1\,]=0 & \text{on $\{y=0\}$.}
\end{cases}
\end{equation}

 We now consider the canonical decomposition  $H=e_0 Du_h+K_{h, \s}$, where $u_h$ and $K_{h, \s}$ are defined as in \eqref{eleq} and \eqref{Kcanonical}, respectively. Moreover, we decompose also $\dot{H}$ as $\dot{H}=Dv+K$, where 
\begin{equation}\label{201}
K:=
\left(
\begin{array}{cc}
-D_{yy}w_1 & D_{xy}w_1\\
-D_{yy}w_2 & D_{xy}w_2
\end{array}
\right)
\end{equation}
with $w=(w_1, w_2)$ the unique solution in $H^1_\#(\Om_h; \R^2)$ to
$$
\begin{cases}
\Delta w= -\bu_j \ro_{r_0}(\cdot-z_j) & \text{in $\Om_h$,}\\
w=0 & \text{on $\Gamma_h$,}\\
w=0 & \text{on $\{y=0\}$.}\
\end{cases}
$$
 We note that since $D_{xx}w=0$  and $\ro_{r_0}((x,0)-z_j)=0$ on $\{y=0\}$ (the last condition comes from the fact that $B_{r_0}(z_j)\subset\Om_h$), from the equation satisfied by $w$ we deduce that
$D_{yy}w=0$ on $\{y=0\}$, which in turn implies that $K[\, \mathbf{e}_1\, ]=-Dv[\,\mathbf{e}_1\,]=0$. Thus, $v$ can be chosen to be identically zero on $\{y=0\}$.
Then, by \eqref{eleq} we have
\beq\label{schifo0}
\begin{array}{l}
\displaystyle\int_{\Om_h} \C H_{sym}: \dot{H}_{sym}\, dz =e_0 \int_{\Om_h}\C E(u_h):\dot{H}_{sym}\, dz +\int_{\Om_h}\C (K_{h, \s})_{sym}:\dot{H}_{sym}\, dz\vspace{10pt}\\
\displaystyle= e_0 \int_{\Om_h}\C E(u_h): E(v)\, dz +e_0 \int_{\Om_h}\C E(u_h):K_{sym}\, dz+\int_{\Om_h}\C (K_{h, \s})_{sym}:\dot{H}_{sym}\, dz\vspace{10pt}\\
\displaystyle= e_0 \int_{\Om_h}\C E(u_h):K_{sym}\, dz+\int_{\Om_h}\C (K_{h, \s})_{sym}:\dot{H}_{sym}\, dz\vspace{10pt}\\
\displaystyle=e_0 \int_{\Om_h}\C E(v_0):K_{sym}\, dz+ e_0 \int_{\Om_h}\bigl(\C E(u_h)-\C E(v_0)\bigr):K_{sym}\, dz\vspace{10pt}\\
\displaystyle \quad+\int_{\Om_h}\C (K_{h, \s})_{sym}:\dot{H}_{sym}\, dz.
\end{array}
\eeq
By \cite[Lemma~6.10]{FFLM2} for every $\e>0$  there exists $\de>0$ such that  $\|u_h-v_0\|_{C^{1,\alpha}_\#(\Om_h;\R^2)}\leq \e$,
where $v_0$ is defined in \eqref{v0}. 
Hence,
\beq\label{schifo1}
\int_{\Om_h}\bigl|\bigl(\C E(u_h)-\C E(v_0)\bigr):K_{sym}\bigr|\, dz\leq C\e
\eeq
for some positive constant $C$ independent of $e_0$. Observe now that, using \eqref{cxi}, \eqref{v0}, \eqref{200}, and \eqref{201}, 
we have
\begin{align*}
&\int_{\Om_h}\C E(v_0):K_{sym}\, dz=-\frac{4\mu(\mu+\lambda)}{2\mu+\lambda}\int_{\Om_h}D_{yy}w_1\, dz\\
&\qquad\qquad=-\frac{4\mu(\mu+\lambda)}{2\mu+\lambda}\biggl[\int_{\Om_h}\Delta w_1\, dz-\int_{\Gamma_h}D_xw_1(\nu\cdot\mathbf{e}_1)\,d\H^1(z)\biggr]\\
&\qquad\qquad\geq\frac{4\mu(\mu+\lambda)}{2\mu+\lambda}\biggl[(\bu_j\cdot\mathbf{e}_1)\int_{\Om_h}\ro_{r_0}(z-z_j)\, dz-\ell\|Dw_1\|_{L^\infty(\Omega_h;\R^2)}\|h-d/\ell\|_{C^{1,\alpha}_\#(0,\ell)}\biggr]\,,
\end{align*}
where the second equality is due to the fact that $D_xw_1$ is $\ell$-periodic in the $x$-direction. From the above inequality, recalling \eqref{schifo-1}, \eqref{schifo0}, \eqref{schifo1}, and the assumption on $h$ we get
\begin{multline*}
\frac{d}{ds}F(h,  \sigma_s, H_{h, \s_s})_{\bigl|_{s=0}}>
e_0\biggl[\frac{4\mu(\mu+\lambda)}{2\mu+\lambda}(\bu_j\cdot\mathbf{e}_1)\int_{\Om_h}\ro_{r_0}(z-z_j)\, dz-C(\e+\delta) \biggr]\\
+ \int_{\Om_h}\C (K_{h, \s})_{sym}:\dot{H}_{sym}\, dz\,,
\end{multline*}
for a possibly larger constant $C$ depending on the $L^\infty$ bounds on $Dw_1$, hence on the $C^{1,\alpha}$ norm of $h$.
Claim \eqref{derG} follows by taking $\e$ small enough and $e_0$ large enough. Indeed,  by Lemma~\ref{lm:ellcan}, \eqref{Kcanonical}, and \eqref{200}, 
$$
\biggl| \int_{\Om_h}\C (K_{h, \s})_{sym}:\dot{H}_{sym}\, dz\biggr|\leq C |\mathbf{b_j}|\|D_y\ro_0\|_{L^2(\R^2)}\|\sigma*\ro_{r_0}^\#\|_{L^2(\Om_h; \R^2)}\,,
$$
where $C$ is a constant depending only on the Lipschitz constant of $h$.
\end{proof}
\begin{remark}\label{rm:orient}
It can be shown that when $|e_0|$ is sufficiently large dislocations with Burgers vectors $\bu$  satisfying
$$
e_0(\bu\cdot \mathbf{e}_1)>0
$$
are energetically favorable compared to dislocations having the same centers but opposite Burgers vectors, see Corollary~\ref{cor:orient}. 
\end{remark}

In the next theorem we show that for suitable choices of the parameters global minimizers must have all the dislocations lying on the film/substrate interface.
\begin{theorem}\label{prop:exnuova}
Assume $\B\neq\emptyset$, fix $d>2r_0 \ell$ and let $|e_0|>\bar e$, where $\bar e$ is as in  Proposition~\ref{prop:nicebottom}.  Assume also $e_0 (\bu_j\cdot\mathbf{e}_1)>0$ for all $\bu_j\in \mathbf{B}$. Then there exists $\bar \gamma$ such that if $\gamma> \bar \gamma$   any global minimizer $(\bar{h}, \bar{\sigma}, \bar{H})$
of the problem \eqref{minG}  has all dislocations lying at the bottom of $\Om_h$, i.e., $\bar{\sigma}=\sum_{i=1}^k\bu_i\delta^\#_{z_i}$, where all the centers $z_i$ are of the form $z_i=(x_i, r_0)$.
\end{theorem}
\begin{proof}
It is enough to show that given $\gamma_n\to+\infty$ 
and corresponding global minimizers $(h_n, \sigma_n, H_{h_n,\sigma_n})\in X(e_{0};\mathbf{B})$ of  \eqref{minG} with $\gamma_n$ in place of $\gamma$, then for $n$ sufficiently large the dislocation measures $\sigma_n$ have all the centers lying at the bottom. Note that $(h_n,\sigma_n,H_{h_n,\sigma_n})$ is a global minimizers of 
$$
\min\big\{G_n(h,\sigma, H):\, (h, \sigma, H)\in X(e_0;\mathbf{B}),\, |\Om_h|=d\big\}\,,
$$
where $G_n$ is the rescaled functional
$$
G_n(h,\sigma, H):=\frac{1}{\gamma_n}\int_{\Om_{h}}W(H_{sym})\, dz+\H^1(\Gamma_h)+2\H^1(\Sigma_h)\,.
$$

\noindent{\bf Step 1.} (Uniform convergence to the flat configuration) By the compactness result in \cite[Proposition~2.2 and Lemma~2.5]{FFLM} and the semicontinuity proved in \cite[Lemma~2.1]{BC}, there exist $h\in AP(0,\ell)$ and a  subsequence (not relabeled) such that $h_n\to h$ in 
$L^{1}(0,\ell)$ and
$$
\H^1(\Gamma_h)+2\H^1(\Sigma_h)\leq \liminf_n\bigl(\H^1(\Gamma_{h_n})+2\H^1(\Sigma_{h_n}\bigr))\,.
$$
 Thus, if we consider any $g\in AP(0,\ell)$  such that $|\Om_g|=d$ and  $(\sigma,H)$ such that $(g,\sigma,H)\in X(e_0;\mathbf{B})$
  \begin{align*}
  \H^1(\Gamma_h)+2\H^1(\Sigma_h)&\leq \liminf_n\bigl(\H^1(\Gamma_{h_n})+2\H^1(\Sigma_{h_n}\bigr)\\
 & \leq \liminf_nG_n(h_n,\sigma_n,H_{h_n,\sigma_n})\leq \liminf_nG_n(g,\sigma,H)= \H^1(\Gamma_g)+2\H^1(\Sigma_g)\,.
 \end{align*}
 Therefore $h$ minimizes
 $$
 g\mapsto  \H^1(\Gamma_g)+2\H^1(\Sigma_g)\
 $$
 among all functions in $AP(0,\ell)$ such that $|\Om_g|=d$. Hence $h$ is the flat profile $h\equiv d/\ell$. Note that from the above chain of inequalities, taking $g=d/\ell$, we have in particular that 
 $$
 \ell=\H^1(\Gamma_{d/\ell})=\lim_n \bigl(\H^1(\Gamma_{h_n})+2\H^1(\Sigma_{h_n})\bigr)\,.
 $$
 Up to a subsequence we may assume that $\{\Gamma_{h_n}\cup \Sigma_{h_n}\}$ converge in the Hausdorff metric to some compact connected set $K$. By the compactness result \cite[Proposition~2.2]{FFLM}, we have that, up to a  subsequence (not relabeled),
 $\R^2\setminus\Om_n^\#\to \R^2\setminus (\R\times (0,d/\ell))$ in the Hausdorff metric.  From this convergence  it follows 
 (see the proof of \cite[Lemma~2.5]{FFLM}) that $\Gamma_{d/\ell}\subset K$. 
 Hence, by Go\l\c{a}b's  theorem and observing that $\H^1(\overline{\Gamma_{h_n}\cup \Sigma_{h_n}})=\H^1(\Gamma_{h_n}\cup \Sigma_{h_n})$, we have
$$
\H^1(\Gamma_{d/\ell})\leq\H^1(K)\leq\lim_{n\to\infty}\H^1(\Gamma_{h_n}\cup \Sigma_{h_n})=\H^1(\Gamma_{d/\ell})\,.
$$
Therefore, $\H^1(K\setminus\Gamma_{d/\ell})=0$. Since $K$ is the Hausdorff limit of graphs, for all $x\in[0,\ell]$
the section $K\cap(\{x\}\times\R)$ is connected.  Hence, $K=\overline\Gamma_{d/\ell}$. From this equality and the definition of Hausdorff convergence, we get that $\sup_{[0,\ell]}|h_n-d/\ell|\to 0$ as $n\to\infty$.

\noindent{\bf Step 2.} (Penalization) We now show that there exists $\Lambda$ sufficiently large and independent of $n$ such that every  
minimizer of
\beq\label{minGpenn}
\min\big\{G_n(h,\sigma, H)+\Lambda||\Om_h|-d|:\, (h, \sigma, H)\in X(e_0;\mathbf{B})\big\}
\eeq
satisfies the volume constraint associated with $d$. 
 We argue by contradiction assuming that there exist  sequences $\{\Lambda_m\}$ with $\Lambda_m\to\infty$ and $\{n_m\}$, and 
  minimizers $(g_m, \tau_m, H_{g_m, \tau_m})\in X(e_0; \B)$ of  \eqref{minGpenn} with $n=n_m$ such that $|\Om_{g_m}|\neq d$.
 Arguing as in Step 1 of the proof of Theorem~\ref{sollevamento}, one can show that for $n$ large enough $|\Om_{g_n}|>d$.
 We can now proceed as in Step 2 of the same proof to show that either we can cut out a small region from $\Om_{g_m}$, thus strictly reducing the total energy and contradicting the minimality, or we can show that $g_m\to d/b$  uniformly (see \eqref{g max}) and for every $m$ there exist
 a dislocation ball $B_{r_0}(z_m)$ touching $\Gamma_{g_m}$ at a point of maximum height. In particular, up to  a  subsequence (not relabeled),  $\tau_m\wto \tau$ with $\tau=\sum_{i=1}^k\bu_i\delta^\#_{z_i}$ such that we have $z_j=(x_j, d/\ell-r_0)$  for some $j\in \{1, \dots, k\}$; i.e., the corresponding ball $B_{r_0}(z_j)$ is tangent to $\Gamma_{d/\ell}$.  Note also that $H_{g_m, \tau_m}\wto \bar H$ in
 $L^2_{loc}(\Om_{d/\ell}; \mathbb{M}^{2\times 2})$ with $\curl \bar H=\tau*\ro_{r_0}$ and that $\bar H[\,\mathbf{e}_1\, ]= e_0\mathbf{e}_1$. This can be shown arguing as in the proof of Theorem~\ref{th:existence}.
  Observe now that given $\eta\in (0, d/\ell)$, $\sigma \in  \md(\Om_{d/\ell-\eta}; \mathbf{B})$ and 
  $H\in \Hc{{d/\ell+\eta}}$ such that $\curl H=\s*\ro_0$ in $\Om_{d/\ell+\eta}$ and $H[\, \mathbf{e}_1\,]=e_0\mathbf{e}_1$, since $g_m\to d/b$  uniformly, we have that $g_m(x)\le d/\ell+\eta$ for all 
  $x\in(0,\ell)$ and all $m$ sufficiently large. Hence, by  the minimality of $(g_m, \tau_m, H_{g_m, \tau_m})$ and lower semicontinuity we have
  \begin{align*}
  \int_{\Om_{d/\ell}}W(\bar H_{sym})\, dz& \leq \liminf_m  \int_{\Om_{g_m}}W((H_{g_m, \tau_m})_{sym})\, dz\\
  & \leq 
   \liminf_m \int_{\Om_{g_m}}W(H_{sym})\, dz=\int_{\Om_{d/\ell}}W(H_{sym})\, dz\,.
  \end{align*}
 Since $d>2r_0\ell$, by the arbitrariness of $\eta$, $\sigma$ and $H$ we conclude that $\bar H= H_{d/\ell, \tau}$ and $(\tau, H_{d/b, \tau})$ 
  is  a solution of 
  \begin{multline*}
  \min\biggl\{\int_{\Om_{d/\ell}}W(H_{sym})\, dz:\,H\in\Hc{{d/\ell}},\,\\
    \sigma\in \md(\Om_{{d/\ell}}; \B) \text{ such that }(d/\ell, \sigma, H)\in X(e_0; \mathbf{B})\biggr\}\,,
  \end{multline*}
 which contradicts Proposition~\ref{prop:nicebottom},  since $|e_0|>\bar e$ and  there is at least one dislocation which is not lying on the bottom.
 
 \noindent{\bf Step 3.} ($C^{1}$-convergence) By Step 2 and Lemma~\ref{pallaint}, we deduce that $\Om_{h_n}^\#$ satisfies a uniform interior ball condition with any radius $\ro<\min\{1/\Lambda, r_0\}$ and thus independent of $n$. This property, together with the uniform convergence proved in Step 1, implies that for $n$ large $\Sigma_{h_n}\cup \Sigma_{h_n, c}=\emptyset$. This can be shown arguing as in Step 2 of the proof
 of \cite[Theorem~6.9]{FuMo}. In turn, by Theorem~\ref{th:regolarita}, we deduce that for $n$ sufficiently large $\Gamma_{g_n}^\#$ is
 of class $C^{1,\alpha}$ for all $\alpha\in (0, 1/2)$. We now show that in fact $h_n\to d/\ell$ in $C^1_\#([0,\ell])$. 
 
 To this aim, fix $\ro< \min\{1/\Lambda, r_0\}$.  By Step 1 we have $a_n:=\sup_{x\in [0,\ell)}|h_n(x)-d/\ell|\to 0$. Take now $z=(x, h_n(x))$ and the corresponding ball $B_{\ro}(z_0)\subset  \Om^\#_{h_n}\cup(\R\times(-\infty,0])$ touching $\Gamma_{h_n}$ at $z$.   
If $h_n(x)=d/\ell-a_n$ then $h'_n(x)=0$ since $h_n\geq d/\ell-a_n$. Otherwise, let us set
$\Gamma_n:=\pa B_{\ro}(z_0)\cap \{(x,y):\, y\geq d/\ell-a_n\}$.  Since $a_n\to 0$ we have $\H^1(\Gamma_n)\to 0$. Therefore, since
$z\in \Gamma_n$, the slope of the tangent to $\pa B_{\ro}(z_0)$ at $z$ is bounded by a small constant $\omega(\H^1(\Gamma_n))$, where $\omega$ is a continuity modulus such that $\omega(0+)=0$. This shows that $h'_n\to 0$ uniformly in $[0, \ell]$ as claimed.  

\noindent{\bf Step 4.} ($C^{1,\alpha}$-convergence and conclusion) Write $\sigma_n=\sum_{i=1}^k\bu_i\delta^\#_{z_{i,n}}$, $z_{i,n}=(x_{i,n}, y_{i,n})$. We now decompose $H_{h_n, \sigma_n}=Dv_n+K_n$, where
$$
K_n:=
\left(
\begin{array}{cc}
k_{1,n} & 0\\
k_{2,n} & 0
\end{array}
\right)\,,
$$
with 
$$
k_{l,n}(x,y):=-\sum_{i=1}^k(\bu_i\cdot \mathbf{e}_l)\int_0^y\ro_{r_0}(x-x_{i,n}, t-y_{i,n})\, dt\,, \text{ for $l=1,2$,}
$$
and $v_n$ satisfies
$$
\begin{cases}
\Div\, \C E(v_n)=-\Div\, \C (K_n)_{sym} & \text{in $\Om_{h_n}$,}\\
\C E(v_n)[\nu]=- \C (K_n)_{sym}[\nu] & \text{on $\Gamma_{h_n}$.}
\end{cases}
$$
Since $h'_n\to 0$ uniformly, 
we can argue as in \cite[Theorem~6.10]{FuMo}  to prove that for every $\beta\in (\frac 12,1)$ there exist $C>0$ and  a radius $\bar r>0$, both independent of $n$,   such that for all $z_0\in \Gamma_{h_n}$ and for all $r\leq\bar r$,
$$
\int_{B_r(z_0)\cap\Om^\#_{h_n}}|\nabla v_n|^2\, dz\leq C r^{2\beta}
$$
for $n$ large enough.
In turn, since $K_n$ is smooth this implies that for a possibly larger constant  $C>0$ (still independent of $n$) 
$$
\int_{B_r(z_0)\cap\Om^\#_{h_n}}|H_{ h_n, \s_n}|^2\, dz\leq C r^{2\alpha}
$$
 for all $z_0\in \Gamma_{h_n}$, for all $r\leq\bar r$, and for $n$ large enough. From this estimate, arguing exactly as in Step 3 of Theorem~\ref{th:regolarita}, we deduce that there exists a constant $C>0$ such that for all $n$ sufficiently large
 $$
 \medint_{x_0}^{x_0+r}\Bigl|\bar h_n'(x)- \medint_{x_0}^{x_0+r} h_n'\, ds\Bigr|\, dx
 \leq C r^{\beta-\frac12}
 $$
for all $x\in [0,\ell)$ and $r<\bar r$. By \cite[Theorem~7.51]{AFP}, this implies that $\|h_n\|_{C^{1,\beta-\frac12}_\#([0,\ell])}$ is uniformly bounded for $n$ sufficiently large. By the arbitrariness of $\beta\in (\frac12, 1)$, we have shown that $h_n\to d/\ell$ in 
$C^{1,\alpha}_\#([0,\ell])$ for all $\alpha \in (0,\frac12)$. Recalling the choice of $\bar e$, the conclusion of the theorem follows from
Proposition~\ref{prop:nicebottom}.
 \end{proof}

\section{The nucleation energy}\label{sec:nucl}
In this section we will address the nucleation of dislocations. Fix a finite set $\mathcal{B}^o$ of fundamentals Burgers vectors, which are linearly independent with respect to integer linear combinations; i.e.,
if $\bu^o_1$, $\dots$, $\bu^o_N$ are distinct elements of $\mathcal{B}^o$ such that $\sum_{i=1}^N n_i \bu^o_i=0$, with $n_i\in \Z$, then 
$n_1=\dots=n_N=0$.  Define 
$$
\mathcal{B}:=\Bigl\{\sum_{i=1}^Nm_i\bu^o_i:\, m_i\in \Z,\, \bu^o_i\in\mathcal{B}^o,\,, N\in\N\Bigr\}\,.
$$
For every $\bu\in \mathcal{B}$ we set
$$
\|\bu\|^2_{\mathcal{B}^o}:=\sum_{i=1}^N|m_i||\bu^o_i|^2\,,
$$
where the coefficients $m_i$ are such that $\bu=\sum_{i=1}^Nm_i \bu_i^o$.

Given $h\in AP(0,\ell)$, we now define the admissible dislocation measures in $\Om^\#_h$, by setting
\begin{multline*}
\md(\Om_h):=\\
\biggl\{\sigma\in \mathcal{M}(\Om_h^\#; \R^2):\, \sigma=\sum_{i=1}^{k}\bu_i\delta^\#_{z_i},\, \bu_i\in \mathcal{B}, \, z_i\in \Om_h, \text{ with } B_{r_0}(z_i)\subset\Om_h^\#,\, k\in \N\biggr\}\,.
\end{multline*}
If $\sigma=\sum_{i=1}^k\bu_i\delta^\#_{z_i}\in \md(\Om_h)$, where the $z_i$'s are all distinct,  then the corresponding nucleation energy will be defined as
\begin{equation}\label{ns}
N(\sigma):=c_o\sum_{i=1}^k \|\bu_i\|^2_{\mathcal{B}^o}\,,
\end{equation}
for some (material) constant $c_o>0$.
\subsection{The minimization problem}
For any fixed mismatch strain $e_0\neq 0$ we introduce the space of admissible configurations
\begin{multline*}
X_{e_0}:=\Big\{(h, \sigma, H):\, h\in AP(0,\ell),\, \sigma\in \md(\Om_h),\, H\in \mathbf{H}_\#(\curl; \Om_h; \mathbb{M}^{2\times2})\\
\text{ such that } \curl H=\sigma*\ro_{r_0}\text{ and } H[\,\mathbf{e}_1\, ]=e_0\mathbf{e}_1\Big\}\,,
\end{multline*}
In this section we shall discuss the minimization problem 
\beq\label{minN}
\min\bigl\{F(h, \sigma, H)+ N(\s):\, (h, \sigma, H)\in X_{e_0},\, |\Om_h|=d \bigr\}\,,
\eeq
where $F$ is defined as in \eqref{sharp model energy} and $d>0$ is the given total mass. 
We start by observing that the minimization problem has a solution.
\begin{theorem}\label{th:existenceN}
The minimization problem \eqref{minN} admits a solution.
\end{theorem}
\begin{proof}
Let $\{(h_n, \sigma_n, H_n)\}\subset X_{e_0}$ be a minimizing sequence. 
Note that since $\sup_n N(\sigma_n)<\infty$ and $\min\{\|\bu\|_{\mathcal{B}^o}:\, \bu\in \mathcal{B}\setminus\{0\}\}>0$, we have that the number $k_n$ of centers of the dislocation measures $\sigma_n=\sum_{i=1}^{k_n}\bu_{i,n}\delta^\#_{z_{i,n}}$ is uniformly bounded and 
$\sup_{i,n}\|\bu_{i,n}\|_{\mathcal{B}^o}<+\infty$.  Moreover, arguing as in the proof of Theorem~\ref{th:existence} we have, up to a subsequence, that 
\begin{itemize}
\item[i)] $h_n\to h$ in $L^1(0,\ell)$;
\item[ii)] $\R^2\setminus\Om^\#_{h_n}\to \R^2\setminus\Om^\#_h$ in the sense of the Hausdorff metric,
\end{itemize}
for some $h\in AP(0,\ell)$. 
Therefore, up to extracting a further  subsequence (not relabeled), if needed,  we can assume that there exists 
$k\in \N$ such that 
$\sigma_n=\sum_{i=1}^{k}\bu_{i,n}\delta^\#_{z_{i,n}}$, where $\bu_{i,n}\to \bu_i\in \mathcal{B}$ and $z_{i,n}\to z_i\in \Om_h$, with $B_{r_0}(z_i)\subset\Om_h^\#$. Setting $\sigma=\sum_{i=1}^k\bu_i\delta^\#_{z_{i}}$ and observing that 
$$
N(\sigma)\leq \liminf_n N(\s_n)\,,
$$
we may now conclude arguing exactly as in the proof of Theorem~\ref{th:existence}.
\end{proof}
\begin{remark}[Regularity]  Let $(\bar h,\bar \sigma, H_{\bar h,\s})\in  X_{e_0}$ be a minimizer of problem \eqref{minN}. 
Writing $\bar\s=\sum_{i=1}^k\bu_i\delta^\#_{z_i}$, with $z_i\neq z_j$ if $i\neq j$, set $\B:=\{\bu_1, \dots, \bu_k\}$.  Observe that 
 $(\bar h,\bar \sigma, H_{\bar h,\bar\s})\in X(e_0; \B)$ is also a minimizer of \eqref{minG}. Therefore the regularity 
 Theorem~\ref{th:regolarita} applies.
\end{remark}

\subsection{Existence of configurations with non trivial dislocations}
We start by fixing a  profile $h$ and considering a  minimizer $( \sigma, H_{h, \s})$ of the corresponding energy, i.e.,   $(h, \sigma, H_{h,\s})\in X_{e_0}$ and 
\begin{multline}\label{minel}
\int_{\Om_h}W((H_{h, \s})_{sym})\, dz+ N(\s)\\
=\min\biggl\{\int_{\Om_h}W(H_{sym})\, dz+ N(\tau):\, (\tau, H)\text{ s.t. } (h,\tau, H)\in X_{e_0}\biggr\}\,.
\end{multline}
 We want to show that if $e_0$ is large enough and $h$ is nearly flat, then any minimal configuration $(\s, H_{h,\s})$ has a nontrivial dislocation measure $\s$ and its total variation blows up as $|e_0|\to \infty$.
\begin{proposition}\label{prop:ex1}
Assume that $\mathcal{B}^o$ contains a vector $\bu$ such that $ \bu\cdot \mathbf{e}_1\neq 0$. For every $d>2r_0b$, $M\geq0$, and  $\alpha\in (0,1)$ there exist $\overline e>0$ and $\de>0$ such that if 
$|e_0|>\overline e$, $h\in AP(0,\ell)$ and $\|h-d/\ell\|_{C^{1,\alpha}_\#(0,\ell)}\leq\de$, then for every  minimizer $( \s,H_{h,\s})$ of \eqref{minel}, the dislocation measure $\s$ is nontrivial and the total variation $|\sigma|(\Om_h)>M$. 
\end{proposition}
\begin{proof} We only treat the case $e_0>0$.
 Assume that $|\sigma|(\Om_h)\leq M$. We want to show that if $e_0$ is large enough, this leads to a contradiction.
 Fix $z_0=(x_0, y_0)\in \Om_h$ and consider the dislocation $\overline\s:=\s+ \bu\de_{z_0}^\#\in \md(\Om_h)$ for some $\bu\in \mathcal{B}$ such that  $\bu\cdot\mathbf{e}_1>0$. Such a vector exists by our assumption on $\mathcal{B}^o$.

We consider the canonical decomposition of  $H_{h,\s}$, i.e,   $H_{h,\s}=e_0Du_{h}+K_{h,\s}$, where $K_{h,\s}$ is the unique $\ell$-periodic solution to the system
$$
\begin{cases}
\curl K_{h,\s}=\s*\ro_{r_0} & \text{in $\Om_h$,}\\
\Div\, \C(K_{h,\s})_{sym}=0& \text{in $\Om_h$,}\\
\C(K_{h,\s})_{sym}[\nu]=0 &\text{on $\Gamma_h$,}\\
K_{h,\s}[\, \mathbf{e}_1\,]=0 & \text{on $\{y=0\}$,}
\end{cases}
$$
and $u_h$ is the elastic equilibrium in $\Om_h$ satisfying $u_h(x,0)=(x,0)$. Observe that by \cite[Lemma~6.10]{FFLM2} for every $\e>0$  there exists $\de>0$ such that 
\beq\label{ellreg}
\|h-d/\ell\|_{C^{1,\alpha}_\#(0,\ell)}\leq\de\Longrightarrow \|u_h-v_0\|_{C^{1,\alpha}_\#(\Om_h)}\leq \e\,,
\eeq
where $v_0$ is defined in \eqref{v0}.
 Write $\bu=(b_1, b_2)$ and consider the strain field $e_0Du_{h}+K_{h,\s}+K$,
 where 
 $$
K:=
\left(
\begin{array}{cc}
k_1 & 0\\
k_2 & 0
\end{array}
\right)\,,
\qquad\text{with } k_i(x,y):=-b_i\int_0^y\ro_{r_0}(x-x_0, t-y_0)\, dt\,, \text{ for $i=1,2$.}
$$
Note that by construction $\curl K=\bu\,\delta_{z_0}^\#*\ro_{r_0}$ and $K[\, \mathbf{e}_1\, ]=0$ on $\{y=0\}$.

 A simple calculation shows that
 \begin{align*}
 & \int_{\Om_h}W((H_{h,\s})_{sym}+K_{sym})\, dz-
  \int_{\Om_h}W((H_{h,\s})_{sym})\, dz\\
  & =\int_{\Om_h}W(K_{sym})\, dz+\int_{\Om_h}\C(H_{h,\s})_{sym}:K_{sym}\, dz\\
  &=\int_{\Om_h}W(K_{sym})\, dz+\int_{\Om_h}\C(K_{h,\s})_{sym}:K_{sym}\, dz+e_0\int_{\Om_h}\C E(u_{h}):K_{sym}\, dz\\
  &=\int_{\Om_h}W(K_{sym})\, dz+\int_{\Om_h}\C(K_{h,\s})_{sym}:K_{sym}\, dz+e_0\int_{\Om_h}\C E(v_0):K_{sym}\, dz\\
  &\quad +
  e_0\int_{\Om_h}\bigl(\C E(u_{h})-\C E(v_0)\bigr):K_{sym}\, dz.
  \end{align*}
   Observe that  $\|\sigma*\ro_{r_0}\|_{L^2(\Om_h; \R^2)}\leq C$, where $C=C(M)$ is a constant depending only on $M$. Therefore,  Lemma~\ref{lm:ellcan} implies that 
  $$
  \|K_{h,\s}\|_{L^2(\Om_h; \mathbb{M}^{2\times2})}\leq C\|\sigma*\ro_{r_0}\|_{L^2(\Om_h; \R^2)}\leq C(M)\,.
  $$
  Moreover, we clearly have 
  $$
 N(\overline \s)-N(\s)\leq C\,, 
  $$
  for a possibly different constant depending on  $\mathbf{b}$. Thus, since $\bu\cdot\mathbf{e}_1>0$ we have
  $$
  \int_{\Om_h}\C E(v_0):K_{sym}\, dz=\frac{4\mu(\mu+\lambda)}{2\mu+\lambda}\int_{\Om_h}k_1\, dz<0\,.
  $$
  Hence, also by \eqref{ellreg}, we conclude that there exist two positive constants $c_1$ and $c_2$ (depending only on $d$, $M$, $\mathbf{b}$,  and the Lam\'e coefficients) such that 
  \begin{align*}
  \int_{\Om_h}W((H_{h,\s})_{sym}+K_{sym})\, dz)+ N(\bar \s)-& 
  \int_{\Om_h}W((H_{h,\s})_{sym})\, dz- N(\s)\\
 & \leq 
  c_1+e_0
  \frac{4\mu(\mu+\lambda)}{2\mu+\lambda}
 \int_{\Om_h}k_1\, dz+c_2e_0\|u_h\!-\!v_0\|_{C^{1,\alpha}_\#(\Om_h)}
  \\
  &<c_1+e_0\biggl(\frac{4\mu(\mu+\lambda)}{2\mu+\lambda}\int_{\Om_h}k_1\, dz+c_2\e\biggr)<0
  \end{align*}
  provided that  $\e$ is sufficiently small and $e_0$ is sufficiently large. This contradicts   the minimality of $(\s, H_{h,\s})$.
\end{proof}

\begin{corollary}\label{cor:orient}
 For every $d>0$, $M>0$, and  $\alpha\in (0,1)$ there exist $\overline e>0$ and $\de>0$ such that if  $|e_0|>\overline e$, $h\in AP(0,\ell)$ and $\|h-d/\ell\|_{C^{1,\alpha}_\#(0,\ell)}\leq\de$ and $\sigma=\sum_{i=1}^k\bu_i\delta^\#_{z_i}\in
 \mathcal{M}_{dis}(\Om_h)$ with $|\sigma|(\Om_h)\leq M$, $e_0(\bu_j\cdot \mathbf{e}_1)<0$ for  $j\in J\subset \{1,\dots, k\}$, $J\neq \emptyset$, then 
 $$
 \int_{\Om_h}W((H_{h,\s})_{sym})\, dz> \int_{\Om_h}W((H_{h, \tilde\sigma})_{sym})\, dz\,,
 $$ 
where 
$$
\tilde \sigma=\sum_{i\not \in J} \bu_i\delta^\#_{z_i}-\sum_{i \in J} \bu_i\delta^\#_{z_i}\,.
$$
\end{corollary}
\begin{proof}
It is enough to show  that the energy strictly decreases whenever we replace $\bu_j$, with $j\in J$, by $-\bu_j$. Indeed, set $\bar \sigma:=\sigma- 2\bu_j \delta^\#_{z_j}$. Arguing exactly as in Proposition~\ref{prop:ex1}, we have that for $|e_0|$ sufficiently large 
\begin{align*}
\int_{\Om_h} W((H_{h, \bar \sigma})_{sym})\, dz & -\int_{\Om_h}W( (H_{h, \sigma})_{sym})\, dz \\
& \leq\int_{\Om_h} W((H_{h,  \sigma}+K)_{sym})\, dz-\int_{\Om_h}W( (H_{h, \sigma})_{sym})\, dz<0\,,
\end{align*}
where 
$$
K:=
\left(
\begin{array}{cc}
k_1 & 0\\
k_2 & 0
\end{array}
\right)\,,
\qquad\text{with } k_i(x,y):=-2(\bu_j\cdot \mathbf{e}_i)\int_0^y\ro_{r_0}(x-x_0, t-y_0)\, dt\,, \text{ for $i=1,2$.}
$$
\end{proof}

As an application of Proposition~\ref{prop:ex1} and of the theory developed in \cite{FuMo}, we show that for suitable values of $e_0$ and $\gamma$
the global minimizers display a nontrivial dislocation part. 
\begin{theorem}[Minimizers with dislocations]\label{th:ul}
Assume that $\mathcal{B}^o$ contains a vector $\bu$ such that $ \bu\cdot \mathbf{e}_1\neq 0$, fix $d>2r_0 \ell$ and let $|e_0|>\bar e$, where $\bar e$ is as in  Proposition~\ref{prop:ex1}.  Then there exists $\bar \gamma$ such that if $\gamma> \bar \gamma$, then    any global minimizer $(\bar{h}, \bar{\sigma}, \bar{H})$
of the problem \eqref{minN}  has nontrivial dislocations, i.e., $\bar{\sigma}\neq 0$.
\end{theorem}

\begin{proof}
Assume without loss of generality that $e_0>\bar e$ and assume by contradiction that there exists a sequence $\gamma_n\to+\infty$ 
and a corresponding sequence $(h_n, \sigma_n, H_n)\in X_{e_{0}}$ of global minimizers for \eqref{minN}, with $\gamma$ replaced by $\gamma_n$, such that  $\sigma_n=0$. In particular $H_n=e_{0,n}D u_{h_n}$, where $u_{h_n}$ is the elastic equilibrium  in $\Om_{h_n}$ (see \eqref{eleq}).  It follows that $(h_n, u_{h_n})$ is a global minimizer
of 
$$
\min\big\{G_n(h,u):\, (h, 0, D u)\in X_{1},\, |\Om_h|=d\big\}\,,
$$
where
$$
G_n(h,u):=\frac{1}{\gamma_n}\int_{\Om_{h}}W(E(u))\, dz+\H^1(\Gamma_h)+2\H^1(\Sigma_h)\,.
$$
Arguing exactly as in Step 1 of the proof of Theorem~\ref{prop:exnuova} we can show that
 $\sup_{[0,\ell)}|h_n-d/\ell|\to 0$. We claim that 
 \beq\label{claim}
 h_n=d/\ell \qquad \text{for $n$ large enough.}
 \eeq
To this aim, we argue by contradiction assuming $\sup_{x\in[0,\ell]}|h_n(x)-d/\ell|>0$ for a (not relabeled) subsequence,  
  Note  that we may rewrite the functional $G_n$
 as
 $$
G_n(h,u):=\int_{\Om_{h}}W_n(E(u))\, dz+\H^1(\Gamma_h)+2\H^1(\Sigma_h)\,,
$$
 where $W_n$ is defined as in \eqref{canonico}, with $\mu$ and $\lambda$ replaced by $\mu_n:=\mu \frac{1}{\gamma_n}$ and 
 $\lambda_n:=\lambda \frac{1}{\gamma_n}$, respectively.
Since $\mu_n\to 0$ and $\lambda_n\to 0$,  we may apply the local minimality result in \cite[Theorem~2.9]{FuMo}, to conclude that there exist
$n_0$ and $\de>0$ such that
\beq\label{locmin1}
G_{n_0}(d/\ell,u_{d/\ell})<G_{n_0}(k,u_k)
\eeq
 whenever $k\in AP(0,\ell)$, $|\Om_k|=d$, and $0<\sup_{x\in[0,\ell]}|k(x)-d/\ell|\leq \de$.
 
 Take $n>n_0$ so large that 
 $$
 0<\sup_{x\in[0,\ell]}|h_n(x)-d/\ell|\leq \de\quad\text{and}\quad 
 \frac{\gamma_{n_0}}{\gamma_n}<1\,.
 $$ 
 From the inequalities above and \eqref{locmin1}, we get
 \begin{align*}
 G_n(d/\ell, u_{d/\ell})& = \frac{\gamma_{n_0}}{\gamma_n}G_{n_0}(d/\ell, u_{d/\ell})
 +\biggl(1- \frac{\gamma_{n_0}}{\gamma_n}\biggr)\H^1(\Gamma_{d/\ell})\\
& < \frac{\gamma_{n_0}}{\gamma_n}G_{n_0}(h_n, u_{h_n})
 +\biggl(1- \frac{\gamma_{n_0}}{\gamma_n}\biggr)\bigl(\H^1(\Gamma_{h_n})+2\H^1(\Sigma_{h_n})\bigr)\\
 &=G_n(h_n, u_{h_n}),
 \end{align*}
 thus contradicting the minimality of $(h_n, u_{h_n})$. This proves claim \eqref{claim}.
In turn, by Proposition~\ref{prop:ex1} we deduce that  for $n$ sufficiently large $\sigma_n\neq 0$, in contrast with our initial contradiction assumption. 
\end{proof}

\section*{Acknowledgements}

The authors wish to acknowledge the Center for Nonlinear Analysis (NSF PIRE
Grant No. OISE-0967140) where part of this work was carried out. The research
of I. Fonseca was partially funded by the National Science Foundation under
Grant No. DMS-1411646 and that of G. Leoni under Grant No. DMS-1412095. The research of N. Fusco and M. Morini was partially funded by  the FiDiPro project 2100002028 of the Finnish Academy of Science. The
authors would like to thank Kaushik Dayal for useful conversations.

\end{document}